\documentclass[twoside,11pt]{article}

\usepackage{jmlr2e}

\usepackage{amsfonts, dsfont, amssymb, amsmath,graphicx,xcolor,algorithm} %
\usepackage{semantic}

\usepackage{algpseudocode, cleveref}

\newcommand{\ignore}[1]{}

\definecolor{forestgreen}{rgb}{0.0, 0.27, 0.13}

\newcommand\numberthis{\addtocounter{equation}{1}\tag{\theequation}}

\newtheorem{defn}{Definition}
\newtheorem{thm}{Theorem}

\newtheorem{cor}{Corollary}[thm]
\newtheorem{lem}{Lemma}

\mathlig{==}{\equiv}
\mathlig{=.}{\doteq}
\mathlig{:=}{\triangleq}
\mathlig{<<}{\ll}
\mathlig{>>}{\gg}
\mathlig{<>}{\neq}
\mathlig{<=}{\leq}
\mathlig{>=}{\geq}
\mathlig{<==}{\Leftarrow}
\mathlig{==>}{\Rightarrow}
\mathlig{<=>}{\Leftrightarrow}
\mathlig{<==>}{\iff}
\mathlig{<--}{\leftarrow}
\mathlig{-->}{\rightarrow}
\mathlig{<->}{\leftrightarrow}
\mathlig{+-}{\pm}
\mathlig{-+}{\mp}
\mathlig{...}{\dots}
\mathlig{!=}{\stackrel{!}{=}}

\DeclareMathOperator*{\argmax}{\arg\!\max}
\DeclareMathOperator*{\argmin}{\arg\!\min}
\newif\ifshowanswer    %
\newcommand{\isitthree}[1]
{
  \ifnum#1=3
    number #1 is 3
  \else
    number #1 is not 3
  \fi
}

\newcommand{\be}{\begin{equation}}
\newcommand{\ee}{\end{equation}}

\newcommand\R{{\mathbb{R}}}

\renewcommand\P{{\mathds{P}}}
\newcommand\E{{\mathds{E}}}

\renewcommand\b{\boldsymbol}

\newcommand\CD{{\mathcal D}}

\newcommand\N{{\mathbb N}}

 \usepackage{subcaption}

\renewcommand{\epsilon}{\varepsilon}

\newcommand{\acost}{\texttt{aCost}}
\newcommand{\ecost}{\texttt{eCost}}
\newcommand{\adepth}{\texttt{aDepth}}
\newcommand{\edepth}{\texttt{eDepth}}
\renewcommand\b[1]{{\boldsymbol{#1}}}

\usepackage{makecell}

\ShortHeadings{Adaptive Data Depth}{Baharav and Lai}
\firstpageno{1}

\begin{document}

\title{Adaptive Data Depth via Multi-Armed Bandits}

\author{\name Tavor Z.\ Baharav \email tavorb@stanford.edu \\
       \addr Department of Electrical Engineering\\
       Stanford University\\
       Stanford, CA 94305, USA
       \AND
       \name Tze Leung Lai \email lait@stanford.edu \\
       \addr Department of Statistics\\
       Stanford University\\
       Stanford, CA 94305, USA}

\editor{}

\maketitle

\begin{abstract}%
Data depth, introduced by \citet{tukey1975mathematics}, is an important tool in data science, robust statistics, and computational geometry.
One chief barrier to its broader practical utility is that many common measures of depth are computationally intensive, requiring on the order of $n^d$ operations to exactly compute the depth of a single point within a data set of $n$ points in $d$-dimensional space.
Often however, we are not directly interested in the absolute depths of the points, but rather in their \textit{relative ordering}.
For example, we may want to find the most central point in a data set (a generalized median), or to identify and remove all outliers (points on the fringe of the data set with low depth).
With this observation, we develop a novel and instance-adaptive algorithm for adaptive data depth computation by reducing the problem of exactly computing $n$ depths to an $n$-armed stochastic multi-armed bandit problem which we can efficiently solve.
We focus our exposition on simplicial depth, developed by \citet{liu1990notion}, which has emerged as a promising notion of depth due to its interpretability and asymptotic properties.
We provide general instance-dependent theoretical guarantees for our proposed algorithms, which readily extend to many other common measures of data depth including majority depth, Oja depth, and likelihood depth.
When specialized to the case where the gaps in the data follow a power law distribution with parameter $\alpha<2$, we show that we can reduce the complexity of identifying the deepest point in the data set (the simplicial median) from $O(n^d)$ to $\tilde{O}(n^{d-(d-1)\alpha/2})$, where $\tilde{O}$ suppresses logarithmic factors.
We corroborate our theoretical results with numerical experiments on synthetic data, showing the practical utility of our proposed methods.
\end{abstract}

\begin{keywords}
  multi-armed bandits, data depth, adaptivity, large-scale computation, simplicial depth
\end{keywords}

\section{Introduction and Background}

For a set $\CD$ of $n$ points in $d$-dimensional space, the simplicial depth of a point $x$ is defined as the fraction of closed $d$ simplices with vertices in $\CD$ which contain $x$.
This empirical measure of depth was formulated by \citet{liu1990notion}, and satisfies many desired properties of depth including affine invariance, and in the continuous limit satisfies maximality in the center and monotonicity relative to the deepest point.
This leads to its usage in robust statistics \citep{maronna2019robust}, computational geometry \citep{becker1987dynamic}, and exploratory data science \citep{wellmann2010tests}.
Dating back to its inception, simplicial depth has been utilized as a metric for detecting outliers in a multivariate data cloud \citep{liu1990notion}.
One common object of interest, particularly in data science settings, is the simplicial median \citep{gil1992geometric}.
This is defined as the point within a data set with the largest simplicial depth, which is equivalent to the traditional median for 1-dimensional data.
The simplicial median satisfies several natural desiderata when extended to high-dimensional data as discussed in \citet{liu1990notion}, which has led to its use in multivariate data analysis \citep{rousseeuw2004computation, schervish1987review}.
The simplicial median falls within a data set, unlike a mean, and so has meaning even if the data is sparse in some domain, as is common in single cell RNA-seq data sets \citep{ntranos2016fast}. %
Additionally, the simplicial median has some level of robustness, with a data dependent breakdown value of up to $1/(d+2)$, and hence is more robust to outliers and adversarial perturbations than objects like the mean \citep{rousseeuw2004computation}.

The rest of the paper is organized as follows:
in the remainder of this section %
we discuss related works, providing background on data depth, multi-armed bandits, and the use of adaptivity as a computational tool.
We formalize the problem setting in \Cref{sec:form}, summarizing our results in \Cref{table:mainResults}.
We construct and provide theoretical guarantees for our adaptive algorithm for simplicial median computation in \Cref{sec:proposedAlg}, which we generalize to other adaptivity amenable tasks in Sections \ref{sec:topk} and \ref{sec:meta}, providing our meta algorithm for adaptive depth computation in \Cref{alg:adaDepthMeta}.
Numerical experiments are provided in \Cref{sec:numerics} to corroborate our theoretical results.
In \Cref{sec:disc} we discuss two computational and theoretical extensions for our proposed algorithms.
We conclude in \Cref{sec:conc}.
Proofs are relegated to \Cref{app:additionalProofs}, and additional experimental details are provided in \Cref{app:expDetails}.

\subsection{Related Work on Data Depth}
The mathematical formalization of data depth dates back to Tukey's halfspace median \citep{tukey1975mathematics}.
Due to the lack of efficient algorithms for its computation, other notions have been proposed, including Oja depth \citep{oja1983descriptive}, Peel depth \citep{barnett1976ordering}, and simplicial depth \citep{liu1990notion}, each with their own strengths and weaknesses \citep{gil1992geometric}.
In this work we focus on simplicial depth, which was introduced by \citet{liu1990notion}, receiving further treatment by \citet{liu1999multivariate}.
Additional proposed improvements regarding the finite-sample characterization of simplicial depth were proposed by \citet{burr2006simplicial}, defining a modified simplicial depth as the average of the fraction of open and closed simplices the point is contained within.
Naively computing the simplicial depth of a query point among $n$ points in $d$-dimensional space requires $O(n^{d+1})$ time as it entails enumerating all possible simplices, but considerable effort has been devoted to developing more efficient geometric methods.
State-of-the-art algorithms require $O(n\log n)$ time for $d=2$ \citep{rousseeuw1996algorithm}, $O(n^2)$ time for $d=3$ \citep{cheng2001algorithms},  and $O(n^{d-1})$ for $d\ge3$ \citep{pilz2020crossing}.
A general lower bound of $\Omega(n^{d-1})$ has been conjectured for exactly computing the simplicial depth of a single point, and while this has been shown to be achievable up to logarithmic factors, this too quickly becomes infeasible for even moderate $n$.
One approach to make this task computationally tractable is to relax the problem and instead only require a $(1+\epsilon)$ approximation of a points simplicial depth. In this case, it was recently shown that the simplicial depth of a given query point could be approximated in $\tilde{O}(n^{d/2+1})$ time for constant $\epsilon>0$ \citep{afshani2015approximating}.
Additive approximations can be achieved by utilizing $\epsilon$-nets and $\epsilon$-approximations \citep{bagchi2007deterministic} as discussed by \citet{afshani2015approximating},
however such approaches yield poor performance as $\epsilon\rightarrow 0$.

While simplicial depth refers to the depth of a point $x$ with respect to a distribution $f$, and sample simplicial depth refers to the depth of a point $x$ with respect to a data set $\CD$ (i.e., the empirical distribution of $\CD$), for readability we omit the qualifier ``sample'' when it is clear from context.

We now focus our attention on the computation of the simplicial median, the point in the data set with the deepest simplicial depth (uniqueness issues discussed later).
This was tackled by \citet{gil1992geometric}, where a scheme was proposed for $d=2$ which utilizes an efficient simplicial depth algorithm similar to that of \citet{rousseeuw1996algorithm} to compute the simplicial depth of each point exactly, requiring $O(n^2 \log n)$ time.
In the same work the authors propose an algorithm for $d=3$, later corrected by \citet{cheng2001algorithms}, which computes the simplicial median in $O(n^3)$ time in a similar manner.
Both of these schemes rely on an efficient algorithm for computing the simplicial depth of a single point, and loop over each point to compute its exact depth using this fast algorithm.
Thus, finding the simplicial median degenerates to $n$ independent simplicial depth computations.
Indeed, it has been conjectured that these algorithms are optimal for the computation of the simplicial median if all the simplicial depths are required to be computed exactly \citep{gil1992geometric}.
While in a minimax sense we can provide no improvement over existing methods, our proposed algorithm is able to provide instance dependent complexity guarantees, where for easier problem instances we are able to only coarsely approximate the simplicial depths of many points, yielding dramatically improved sample complexities.
We mention that other works have tackled more general depth approximation problems, where a scheme for determining a point with deepest simplicial depth over all of $\R^d$ up to a multiplicative $(1+\epsilon)$ was proposed by \citet{aronov2008approximating}. 
Their work optimizes over all of $\R^d$ rather than simply within the data set however, leading to a more general problem and a computational complexity of $\Omega(n^{d+1})$ for their algorithm in this setting.

On a more practical note, we see that no algorithms for $d>2$ are implemented and empirically tested in the literature (approximation or otherwise).
These algorithms rely on complex combinatorial subroutines, and many only yield asymptotic error probability guarantees (for a fixed $\epsilon$) with unspecified or unwieldy constants.
This mitigates their utility in real world applications, where the user wishes to specify both an allowable error tolerance $\epsilon$ (potentially 0), as well as a maximum error probability $\delta$.
This motivates our efficient multi-armed bandit-based algorithm, which is able to satisfy these desiderata in a theoretically sound, computationally efficient, and modular algorithm.

\subsection{Related Work on Multi-Armed Bandits} \label{sec:adaTasks}

The $k$-armed bandit problem was introduced by \citet{robbins1952some} for $k=2$ in his seminal paper on sequential design of experiments, in which he considered sequential sampling from two populations with unknown means to maximize the total expected rewards $\E[y_1 + \hdots + y_n]$, where $y_i$ has mean $\mu_1$ (or $\mu_2$) if it is sampled from population 1 (or 2) and $n$ is the total sample size.
Letting $s_n = y_1 + \hdots +y_n$, he applied the law of large numbers to show that $\lim_{n \to \infty} n^{-1} \E [s_n] = \max(\mu_1,\mu_2)$ is attained by the following rule: sample from the population with the larger sample mean except at times belonging to a designated sparse set $T_n$ of times, and sample from the population with the smaller size at the designated times.
$T_n$ is called ``sparse'' if $\#(T_n) \to \infty$ but $\#(T_n)/n \to 0$ as $n\to \infty$, where $\#(\cdot)$ denotes the cardinality of a set.

In 1957, Bellman introduced the dynamic programming approach to Robbins' 2-armed sequential sampling problem, generalizing it to $k$-arms and calling it a ``$k$-armed bandit problem'' \citep{bellman1966dynamic}.
The name derives from an imagined slot machine with $k$ arms (levers) such that when arm $j$ is pulled the player wins a reward generated by an unknown probability distribution $\Pi_j$.
In this setting where we seek to maximize the expected sum of rewards there is a fundamental dilemma between ``exploration'' (to obtain information about $\Pi_1,\hdots, \Pi_k$ by pulling the individual arms) and ``exploitation'' (of the information so that inferior arms are pulled minimally).
Dynamic programming offers a systematic solution of the dilemma in the Bayesian setting but suffers from the curse of dimensionality as $k$ and $n$ increase, as recognized by Bellman.

Extending these regret analyses beyond the standard stochastic multi-armed bandit setting we consider in this work, \citet{lai1995machine} were able to use large deviation bounds to provide a regret lower bound in the nonparametric setting pioneered by \citet{yakowitz1991nonparametric}, also providing a UCB (upper confidence bound) based adaptive allocation rule.
Using a different method, \citet{Auer2002} developed an $\epsilon$-greedy--based procedure similar to that used in reinforcement learning \citep{sutton2018reinforcement}, which they showed was also able to attain the logarithmic regret lower bound.

These ideas have been further developed in the contextual setting, where \citet{kim2021multi} generalize the definition of regret to the nonparametric setting with stochastic covariates $\{\b{x}_t\}_{t=1}^n$ which are independent and identically distributed.
There, the authors construct an arm elimination based $k$-armed adaptive allocation procedure using $\epsilon$-greedy randomization and a Welch-type test statistic, which generalizes parametric likelihood ratio statistics to the nonparametric setting.
\citet{kim2021multi} show that that under weak regularity conditions their procedure attains the asymptotic minimax rate of the risk functions for adaptive allocation rules.
The authors additionally discuss their advances in the context of reinforcement learning for recommender systems and personalization technologies. 
This motivates a practically important extension of the contextual setting  to the scenario where the number of arms scales with the time horizon and so $k= k_n$ (instead of fixing $k$ and letting $n$ tend to infinity), in which case $k_n$ can exceed $n$, which has been definitively settled in the recent work of \citet{lai2021Bandit} who consider a non-denumerable set of arms.
This recent advance utilizes the decoupling inequalities of \citet{victor2000theory} to handle general dependence structure (unknown) among the covariates.
Related applications of this breakthrough have been given by \citet{sklar2021bandit} and \citet{lai2021novel}.
We direct the interested reader to these works, in particular \citet{lai2021Bandit} and \citet{kim2021multi}, for a more thorough discussion of regret minimization techniques, a broader background regarding the multi-armed bandit literature, and a thorough characterization of the nonparametric contextual multi-armed bandit setting.

Since these foundational works, multi-armed bandits have proven to be an extremely broad and useful model for dealing with decision making under uncertainty, having been utilized in the design of adaptive clinical trials \citep{villar2015multi}, online ad recommendation \citep{lu2010contextual}, multi-agent learning \citep{bistritz2021one}, and many more applications \citep{slivkins2019introduction,bubeck2012regret}.
Most relevant to this work on data depth are the pure exploration variants of the multi-armed bandit problem, where the prototypical objective is to find the best arm \citep{jamieson2014best}, which in this case translates to finding the point with the deepest simplicial depth.
Common algorithmic paradigms for this setting involve either index-based policies similar to the original Upper Confidence Bound (UCB) algorithm \citep{lai1985asymptotically}, and elimination-based algorithms where arms are successively pulled and discarded after they are determined to be suboptimal with high probability \citep{audibert2010best}.
The hardness of the best arm identification problem has been precisely characterized through fundamental lower bounds \citep{kaufmann2016complexity}, algorithms that are within $\log \log$ factors of optimal \citep{jamieson2014lil}, and asymptotically optimal algorithms following the Track-and-Stop paradigm \citep{garivier2016optimal}.

There are many pure exploration problems beyond best arm identification where adaptivity is beneficial.
The most natural extension is to the PAC scenario, where the objective is to identify an arm with mean within $\epsilon$ of the best arm \citep{kalyanakrishnan2012pac}.
Another extension is to the top-$k$ setting, where the objective is  to identify the $k$ arms with largest means \citep{simchowitz2017simulator}.
If a specific value $k$ is not desired, one can instead find all arms with means within an additive or multiplicative $\epsilon$ of the best arm \citep{mason2020finding}.
Another approach would be to coarsely rank the arms into batches \citep{KatJaiEtAl,karpov2020batched}.
Alternatively, one may wish to identify all arms with means above a given input threshold, in what is called the thresholding bandit problem \citep{locatelli2016optimal}.
A final problem could be determining one arm with mean above a threshold in the good arm identification problem \citep{katz2020true}.
While many of these objectives could be of interest in the context of data depth, 
in this work we focus on the task of identifying the simplicial median, corresponding to the best-arm identification problem.

\vspace{-.1cm}
\subsection{Adaptivity as a Computational Tool}
\vspace{-.05cm}
Due to the ever increasing size of data sets, multi-armed bandit-based randomized algorithms have been recognized as a useful tool for constructing instance-optimal algorithms for data science primitives.
From its early uses in Monte Carlo Tree Search \citep{kocsis2006bandit} to more recent utilization in hyper-parameter tuning \citep{li2017hyperband}, adaptivity has been used to focus computational resources on relevant portions of computational tasks.
Recently formalized under the general framework of Bandit-Based Monte Carlo Optimization (BMO) \citep{bagaria2021bandit}, this technique has been utilized to solve many problems, including finding the medoid of a data set \citep{bagaria2018medoids,baharav2019ultra,tiwari2020bandit}, $k$-nearest neighbor graph construction \citep{bagaria2021bandit,lejeune2019adaptive,mason2019learning,mason2021nearest}, Monte Carlo permutation-based multiple testing \citep{zhang2019adaptive}, mode estimation \citep{singhal2020query}, and a rank-one estimation problem \citep{kamath2020adaptive}. 
This general framework provides efficient algorithms for solving problems of the form $f(\theta_1,\hdots,\theta_n)$, where each $\theta_i$ is an expensive to compute but easy to estimate quantity, and $f$ is a function that does not depend too strongly on too many of the $\theta_i$.
This algorithm design philosophy is discussed in more detail in the recent work of \citet{bagaria2021bandit}, and analyzed in the case when $f$ is real-valued and smooth by \citet{baharav2022approx}.
The resulting adaptive algorithms provide a natural and efficient way to solve certain structured large scale computational problems, allowing for dramatic improvements in computational complexity. To realize these theoretical gains as wall-clock improvements, the incorporation of round-efficient multi-armed bandit algorithms is critical.
This allows for batched computations, yielding up to 4-5 orders of magnitude speed ups in wall-clock time \citep{baharav2019ultra}.
In this paper we construct efficient algorithms for simplicial median computation and several related data depth tasks by building on this BMO framework, showing how these techniques can be utilized to construct instance-adaptive algorithms for modern statistical problems like data depth estimation.

\vspace{-.1cm}
\section{Problem Formulation} \label{sec:form}
\vspace{-.05cm}
For a data set $\CD =\{x_i\}_{i=1}^n$ of $n$ points in $d$-dimensional space, we defined simplicial depth in words as the fraction of closed $d$ simplices with vertices in $\CD$ which contain the target point.
Denoting the convex hull (simplex) of points $x_1,\hdots,x_{d+1}$ as $S[x_1,\hdots,x_{d+1}]$, we are able to mathematically define the sample simplicial depth below.
\vspace{-.1cm}
\begin{defn} \label{defn:simDepth}
The sample simplicial depth of a point $x$ with respect to a data set $x_1,\hdots,x_n$ is defined as
\vspace{-.1cm}
\begin{equation} \label{eq:simDepth}
    D_n(x) := {n \choose d+1}^{-1} \hspace{-.5cm}\sum_{1\le i_1<\hdots<i_{d+1}\le n} \hspace{-.5cm}\mathds{1}\left\{ x \in S[x_{i_1},\hdots,x_{i_{d+1}}]\right\}.
\end{equation}
\end{defn}
\vspace{-.1cm}
As discussed, one natural object of interest in data science applications is the simplicial median, the point in the data set with the largest simplicial depth.
This is formalized as a point in the set
\begin{equation}\vspace{-.1cm}
    \argmax_{x_i \in \CD} D_n(x_i), \label{eq:objective}
\end{equation}%
\noindent which we assume to be unique for the rest of this work for simplicity of presentation.
Our proposed algorithms can accommodate standard solutions to lack of uniqueness such as defining the simplicial median as the average of all points contained in the argmax of \eqref{eq:objective} as is done by \citet{rousseeuw2004computation}.

For the computation of \eqref{eq:simDepth}, determining whether a point $x$ is contained within a given simplex $S[x_1,\hdots,x_{d+1}]$ can be seen to be equivalent to verifying whether $x$ can be expressed as a linear combination with nonnegative coefficients of $\{x_i\}_{i=1}^{d+1}$.
These coefficients are known as the (normalized) barycentric coordinates of $x$ with respect to $\{x_i\}$.
Computing these barycentric coordinates entails solving a $d$-dimensional linear system, which can be efficiently performed in $O(d^3)$ time for a given simplex.
The difficulty in computing \eqref{eq:simDepth} arises from the $n \choose {d+1}$ simplices that naively need to be enumerated to compute the exact depth of a point.
We provide further discussion on this point and on the efficiency of batch operations for this goal in \Cref{sec:barycentric}.

\subsection{Results and Contributions}
\vspace{-.05cm}
Our objective in this paper is to construct an algorithm that efficiently computes the simplicial median \eqref{eq:objective} with error probability at most $\delta$.
We provide in \Cref{table:mainResults} a summary of the results in this paper, where $\tilde{O}$ hides logarithmic factors in the gaps $\{\Delta_i\}$ and $n$.

\def\arraystretch{1.3}
\begin{table}[h]
\hspace{-4.38pt}
\begin{tabular}{l|l|l|l}
                                & Gaps      & Task    & Complexity     \\ \hline
\Cref{thm:main}                  & arbitrary   & simplicial median & See \Cref{thm:main}      \\ \hline
\Cref{cor:d23}                    & arbitrary   & simplicial median   & $\tilde{O}\left(\sum_i \min(\Delta_i^{-2}, n^{d-1})\right)$    \\ \hline
\Cref{cor:powerLaw}             & random, $\alpha$ power law & simplicial median  & $\tilde{O}\left( n^{d -(d-1)\alpha/2}\right)$   \\ \hline
\Cref{thm:simplicialTopk}                  & arbitrary   & deepest-$k$ & $\tilde{O}\left(\sum_i \min\left(\big(\Delta_i^{(k)}\big)^{-2}, n^{d-1}\right)\right)$\\ \hline
\makecell{\Cref{thm:metaAlg} \\ \Cref{cor:metaAlg}}  & arbitrary   & general &  See \Cref{thm:metaAlg}
\end{tabular}
\vspace{-.2cm}
\caption{Summary of the results provided in this paper for dynamic data depth computation for $d\ge 1$. Complexity column only displays coarsely bounded results for ease of readability, see formal statements for tighter results. All sample complexities scale linearly in $\log(1/\delta)$. }\label{table:mainResults}
\vspace{-.2cm}
\end{table}

Unlike prior works which utilize the BMO technique \citep{bagaria2021bandit}, an interesting feature of our algorithms for adaptive simplicial depth computation is that they are able to incorporate in a modular manner an optimized exact computation algorithm.
Concretely, many of these previous tasks have $\theta_i$ (the difficult to compute but easy to approximate quantities) that are expressible as sums, where Monte Carlo samples are generated by sampling a random term in this sum and exact computation is performed by iterating over the entire sum.
In this setting Monte Carlo samples are similarly generated by noting that
\begin{equation}
    \E \left[  \mathds{1}\left\{ x \in S[x_{i_1},\hdots,x_{i_{d+1}}]\right\} \right] = D_n(x),
\end{equation}
for $\{x_{i_j}\}_{j=1}^{d+1}$ drawn uniformly at random from subsets of $\CD$ of size $d+1$.
In this setting however, exact computation can be performed much more efficiently than simply enumerating all ${n \choose d+1} = \theta(n^{d+1})$ simplices (terms in the sum) to exactly compute the simplicial depth of a point; we know from the works of \citet{rousseeuw1996algorithm} and \citet{pilz2020crossing} that exact computation algorithms exist that require only $\tilde{O}(n^{d-1})$ time.
This leads to a reduced threshold for exact computation, the number of Monte Carlo samples beyond which it is determined to be more efficient to exactly compute the depth of a given point than to obtain further Monte Carlo samples to approximate it to greater accuracy. 
This extends and highlights the modularity of the BMO framework \citep{bagaria2021bandit}; better approximation methods and exact computation schemes can be seamlessly integrated into a meta algorithm (\Cref{alg:adaDepthMeta}), yielding improved overall performance.

\section{Proposed Algorithm} \label{sec:proposedAlg}
In this section we formulate and discuss \Cref{alg:adaDepth}, which efficiently finds the simplicial median of an input data set $\CD$ with failure probability at most $\delta$.
This algorithm utilizes the multi-armed bandit algorithm \texttt{Multi-Round $\varepsilon-$Arm} (Algorithm 3 of \citet{hillel2013distributed}).
We proceed in rounds, where in round $r$ the goal is to eliminate all $\epsilon_r = 2^{-r}$ suboptimal points.
This is achieved by constructing $t_r$ total random simplices in round $r$ such that we are able to estimate each remaining point's simplicial depth to within an additive $\epsilon_r/2$ with high probability.
A point's depth is iteratively approximated to greater and greater accuracy until it is either eliminated, or the cost of approximating its depth to accuracy $\epsilon_r/2$ is greater than the cost of exactly computing its simplicial depth, \ecost, as in Line \ref{line:exactCompAlg}.
Note that this cost \ecost \ should be compared in round $r$ with the cost of performing $t_r$ barycentric coordinate computations.
This allows for potentially dramatic computational gains on an instance by instance basis, as the simplicial depth of each point is only approximated to the necessary accuracy for the task at hand (determining if it is the deepest point or not).
\begin{algorithm}[t]
  \caption{\texttt{Adaptive Simplicial Median} \label{alg:adaDepth}}
\begin{algorithmic}[1] 
    \State \textbf{Input: } data set $\CD = x_1,\hdots,x_n\in\R^d$, error probability $\delta$, exact computation method \edepth\ with runtime $\ecost$
    \State \textbf{Initialize: } $S_1 \gets [n]$, $r=0$
    \Repeat
    \State $r\gets r+1$
    \State Set $\epsilon_r = 2^{-r},t_r=\lceil2\epsilon_r^{-2}\log(4nr^2/\delta) \rceil$
    \If{$t_r > \ecost$} \Comment{Use specialized exact method} \label{line:exactCompAlg}
    \State $\hat{\mu}_i^r \gets \edepth(x_i,\CD)$ for $i \in S_r$
    \State \Return $\argmax_{i\in S_{r}} \hat{\mu}_i^r$
    \EndIf
    \State Select $X_1,\hdots,X_{t_r - t_{r-1}}$ independently and u.a.r. as subsets of size $d+1$ of $\CD$
    \State For $i\in S_r$, $j\in[t_{r-1},t_r]$, compute $Y_{i,j}=\mathds{1}\{x_i\in \Delta(X_{j})\}$
    \State For $i \in S_r$ construct $\hat{\mu}_i^r$ from all previous $Y_{i,j}$, let $\hat{\mu}_{*}^r = \max_{i\in S_r} \hat{\mu}_i^r$
    \State Set $S_{r+1}\gets S_r \setminus \{i \in S_r : \hat{\mu}_i^r < \hat{\mu}_{*}^r - \epsilon_r\} $
  \Until $|S_{r+1}| =1$
  \State \Return Point in $S_{r+1}$
\end{algorithmic}
\end{algorithm}

\subsection{Theoretical Guarantees}
We are able to provide sample complexity bounds regarding the performance of Algorithm \ref{alg:adaDepth} by showing that every suboptimal point will be eliminated after an appropriate number of rounds (after its depth has been approximated to the necessary accuracy).
For notational simplicity we assume that the points are sorted in order of simplicial depth and so $x_1$ is the simplicial median, but that the algorithm does not know this.
We define $\mu_i=D_n(x_i)$ and $\Delta_i = \mu_1-\mu_i$, with the understanding that if for some $i$ we have that $\Delta_i=0$, then $\Delta_i^{-2} = +\infty$.
With this definition in hand, we are able to state the following Theorem regarding the number of barycentric coordinate computations made by \Cref{alg:adaDepth} (computations of whether a point is inside a $d$-dimensional simplex).

\begin{thm}[Main Result] \label{thm:main}
Algorithm \ref{alg:adaDepth} succeeds with probability at least $1-\delta$ in returning the simplicial median of an input data set $\CD$, requiring computation at most
\begin{equation*}
n + \sum_{i=1}^n \min\left(\frac{32 \log \left(\frac{16n}{\delta}\log_2^2 \left(\frac{2}{\Delta_i}\right)\right)}{\Delta_i^2}, 2\times \ecost\right)
\end{equation*}
where \ecost\ is the cost of the input exact computation method.
On this $\delta$ probability error event, the algorithm will still terminate within $2 n\times \ecost$ time.
\end{thm}
We analyze the requisite sample complexity of our adaptive algorithm following standard multi-armed bandit techniques.
We begin by showing that each depth estimator $\hat{\mu}_i^{r}$ is within $\epsilon_r/2$ of the true simplicial depth of point $x_i$, $\mu_i$, for all $i,r$ simultaneously with probability at least $1-\delta$.
We then argue that the simplicial median cannot be eliminated on this good event, and that all suboptimal points will be efficiently eliminated.
The proof of this theorem can be found in \Cref{app:additionalProofs}.

We note that the factor of $2$ in front of the \ecost \ term can be made arbitrarily close to 1 by lowering the threshold for exact computation.
This improves the leading constant of the worst case sample complexity of our algorithm but reduces its adaptivity, hampering its instance dependent performance.
For example, if the exact simplicial depth of a point is computed if $t_r \ge c \times \ecost$ for $c\in(0,1]$ in \Cref{line:exactCompAlg} in \Cref{alg:adaDepth}, then the computation done for a point is at most $(1+c)\times \ecost$.
The gap dependent term will naturally suffer however, as points with larger gaps will now also be exactly computed. Up to logarithmic factors, all point with $\Delta_i^{-2} > c\times \ecost$ will have their means exactly computed.

Using the modularity of our proposed algorithms and incorporating the optimized geometric exact computation methods developed by \citet{rousseeuw1996algorithm} and \citet{cheng2001algorithms}, we have the following corollary. %
\begin{cor} \label{cor:d23}
Algorithm \ref{alg:adaDepth} succeeds with probability at least $1-\delta$ in returning the simplicial median of the data set $\CD$, requiring time at most
\begin{equation*}
O\left(\sum_{i=1}^n \min\left(\frac{\log \left(\frac{n}{\delta} \log^2 \left(\frac{1}{\Delta_i}\right)\right)}{\Delta_i^2}, n \log(n) \right)\right),
\end{equation*}
for $d=2$ by utilizing the exact computation scheme in \citep{rousseeuw1996algorithm}, and
\begin{equation*}
O\left(\sum_{i=1}^n \min\left(\frac{\log \left(\frac{n}{\delta} \log^2 \left(\frac{1}{\Delta_i}\right)\right)}{\Delta_i^2}, n^{d-1} \right)\right),
\end{equation*}
for $d\ge3$ by utilizing the exact computation scheme in \citep{pilz2020crossing}.
\end{cor}
These results can be extended to the additive and multiplicative approximation settings, where a point $x_i \in \CD$ is desired such that $D_n(x_i) \ge \max_{j \in [n]} D_n(x_j)-\epsilon$  (respectively $D_n(x_i) \ge (1-\epsilon)\max_{j \in [n]} D_n(x_j)$), by modifying our algorithm's termination condition.

Examining the requisite sample complexity of \Cref{alg:adaDepth}, we see that \Cref{thm:main} provides a general guarantee for all data sets in an instance dependent manner, where all constants are explicit.
To gain insight as to how this computational complexity compares with existing algorithms (which have a sample complexity independent of the gaps) we evaluate the guarantee of \Cref{thm:main} when the gaps between the simplicial depths of the points follow a power law distribution.
We show that this is experimentally a reasonable assumption via experiments on synthetic data sets in \Cref{fig:2dGapHist}.

\clearpage
\begin{cor} \label{cor:powerLaw}
Assume that a data set $\CD = \{x_i\}_{i=1}^n \subset \R^d$ is randomly generated such that $\Delta_i \underset{i.i.d.}{\sim} \Delta$, where $F(\Delta) = \Delta^\alpha$ for $\Delta \in [0,1]$, with constant $\alpha \in [0, \infty)$.
Given this data set, \Cref{alg:adaDepth} will with probability at least $1 - \delta$ identify the simplicial median of the data set, requiring $M$ computations on this success event where, taking the expectation with respect to these $\Delta_i$, we have
\begin{equation}
    \E\{M\} \le 
    \begin{cases} 
      O\left(n\log \left(nd/\delta\right) \ecost^{1-\alpha/2}\right), & \textnormal{for }\alpha \in [0,2), \\
      O\left(n \log \left(nd/\delta\right) \log (\ecost) \right), & \textnormal{for } \alpha = 2, \\
      O\left(n \log \left(nd/\delta\right) \right), & \textnormal{for } \alpha > 2. 
   \end{cases}
\end{equation}
\end{cor}
The proof of this corollary follows by evaluating the sample complexity in \Cref{thm:main} with respect to the random gaps, and is relegated to \Cref{app:powerLaw}.

This shows the clear dependence of our algorithm's sample complexity on the distribution of the gaps.
We see that if many arms gaps are small, corresponding to small $\alpha$, then the algorithm's runtime is near $\ecost^{1-\alpha/2}$ per point in expectation; almost that of the brute force method that exactly computes each simplicial depth.
As $\alpha$ increases our per arm cost decreases, as many arms become easier to eliminate, up until the cost becomes independent of $\alpha$ when $\alpha>2$, as then so few points require exact computation that the dependence on \ecost\ vanishes.

\subsection{Extension to Coarse Ranking} \label{sec:topk}
In practice, one may be interested in finding more than just the simplicial median.
One broad class of objectives of interest are coarse ranking problems, where given as input $\{m_i\}$ where $0=m_0<m_1<\ldots<m_\ell=n$, we wish to coarsely cluster the points by their simplicial depth.
This entails partitioning the $n$ points into $\ell$ clusters, such that the first cluster contains the $(m_1-m_0)$ points with deepest depth, the second cluster contains the following $(m_2-m_1)$ in terms of depth, and so on.
Note that this formulation subsumes simplicial median identification $(\ell=2,m_1=1)$, $k$-outlier detection $(\ell=2,m_1=n-k)$, and sorting the points by simplicial depth $(\ell=n,m_i=i \text{ for } 1\le i \le n)$. 
Round and sample efficient algorithms were recently proposed for the traditional stochastic multi-armed bandit formulation of this problem by \citet{karpov2020batched}.
This has ties to earlier work in the simulation optimization literature, where works exploited the fact that the task of ranking points can be significantly cheaper than approximating all of their depths to high accuracy \citep{ho1992ordinal}.
At a high level, these coarse ranking algorithms maintain a set of active arms and proceed in rounds of doubling accuracy, pulling each active arm uniformly.
When an arm has been pulled sufficiently such that its confidence interval does not overlap with the confidence intervals of any arm in an alternative cluster (i.e. it has been uniquely identified as belonging to a certain cluster), it is removed from the active set (all arms start as active).
For the sake of simplicity and concreteness, in the remainder of this section we consider the objective of finding the $k$ points with largest simplicial depth, but these ideas naturally extend to the general coarse ranking setting.

\clearpage
\begin{thm}[Top-$k$ data depth] \label{thm:simplicialTopk}
Algorithm \ref{alg:adaDepthAlpha} will succeed with probability at least $1-\delta$ in returning a set $A$ with $|A| =k$ such that $\{i : \mu_i > \mu_{k}+\epsilon\} \subseteq A \subseteq \{i : \mu_i \ge \mu_{k} - \epsilon\}$, requiring time at most
\vspace{-.2cm}
\begin{equation*}
O\left( \sum_{i=1}^n \min\left(\frac{\log \left( \frac{n}{\delta} \log \left(1/\Delta_i^{(k)}\right)\right)}{\left(\Delta_i^{(k)}\right)^2}, \ecost \right)\right)
\end{equation*}
\vspace{-.2cm}
where $\Delta_i^{(k)} = \max(\mu_i - \mu_{k+1},\epsilon)$ if $i\le k$ and $\Delta_i^{(k)} = \max(\mu_k - \mu_{i},\epsilon)$ if $i> k$.
\end{thm}
The proof of this theorem follows similarly to that of Theorem \ref{thm:main}, and is relegated to \Cref{app:topk} for continuity.

\begin{algorithm}[t]
  \caption{\texttt{Adaptive top-$k$ Data Depth} \label{alg:adaDepthAlpha}}
\begin{algorithmic}[1] 
    \State \textbf{Input: } data set $\CD = x_1,\hdots,x_n\in\R^d$, error probability $\delta$, target accuracy $\epsilon$, integer $k$,  exact computation method \edepth\ with runtime $\ecost$
    \State \textbf{Initialize: } $S_1 \gets [n]$, $r=0$,  $A\gets \emptyset$
    \Repeat
    \State $r\gets r+1$
    \State Set $\epsilon_r = 2^{-r},t_r=\lceil2\epsilon_r^{-2}\log(4nr^2/\delta) \rceil$
    \If{$t_r \ge \ecost$} \Comment{Use specialized exact method}
    \State $\hat{\mu}_i^r \gets \edepth(x_i,\CD)$ for $i \in S_r$
    \State \textbf{Break: go to line \ref{line:topKIf}}
    \EndIf
    \State Select $X_1,\hdots,X_{t_r - t_{r-1}}$ independently and u.a.r. as subsets of size $d+1$ of $\CD$
    \State For $i\in S_r$, $j\in[t_r]$, compute $Y_{i,j}=\mathds{1}\{x_i\in \Delta(X_{j})\}$
    \State For $i \in S_r$ construct $\hat{\mu}_i^r$ from all previous $Y_{i,j}$, let $\hat{\mu}_\text{thresh}^r = \hat{\mu}_{(k - |A|)}^r$
    \State Update $A \gets A \cup \{i \in S_r : \hat{\mu}_i^r \ge \hat{\mu}_\text{thresh}^r + \epsilon_r\}$
    \State Set $S_{r+1}\gets S_r \setminus \{i \in S_r : |\hat{\mu}_i^r - \hat{\mu}_\text{thresh}^r| \ge \epsilon_r\} $
  \Until $\epsilon_r\le \epsilon/2$ or $|A| =k$
  \If{$|A| < k$} \label{line:topKIf}
  \State Add $k - |A|$ top $\hat{\mu}_i^r$ for $i \in S_r$
  \EndIf
  \State \Return $A$
\end{algorithmic}
\end{algorithm}

While we state \Cref{thm:simplicialTopk} in terms of finding the $k$ points with deepest simplicial depth, similar theoretical guarantees and efficient algorithms can be provided for similar identification and estimation tasks (e.g. identifying the $k$ points with smallest simplicial depth, the outliers of a data set) or partitioning the data into sets (deepest 10\% of points and shallowest 5\% of points).
More generally, considering the order statistics of the sample simplicial depths of the $n$ points, these methods allow for efficient identification or coarse ranking of the points in terms of their simplicial depths. Variable sized outputs can also be efficiently obtained (identifying all points with simplicial depth at least $80\%$ of that of the simplicial median), as discussed in \Cref{sec:adaTasks}.
These same techniques can also be used to estimate L-statistics, efficiently so if the desired linear combination is sparse.
We discuss a meta-algorithm to accomplish these tasks in the following subsection.

\subsection{Meta Algorithm} \label{sec:meta}

At a high level, the previous two proposed algorithms can be abstracted as in \Cref{alg:adaDepthMeta}, which encompasses many other common objectives.
The goal of this meta algorithm is to adaptively estimate the simplicial depth of each point to only the \textit{necessary} accuracy for the task at hand, allowing it to be much more computationally efficient than a naive algorithm which exactly computes the simplicial depth of each point.
This is accomplished by proceeding in rounds, where in each round a given \adepth\ function is used to approximate the depth of the relevant points $x_i$ to an additive tolerance of $\epsilon_r$ which decreases over the rounds $r$, with error probability at most $\delta/(2nr^2)$.
For approximating simplicial depth, this can naturally be accomplished by randomly subsampling simplices and checking to see if the query point $x_i$ is contained within them, as utilized in \Cref{alg:adaDepth}, or by alternative techniques as those of \citet{afshani2015approximating}.
At the end of round $r$, a set of arms $E_r$ is eliminated.
This should be thought of as the set of arms for which the algorithm has already received enough information regarding their means, and no longer needs to refine their estimates.
In the context of top-$k$ identification, this constitutes the points which have been determined to have too shallow of a simplicial depth so as to with high probability not be in the top-$k$, and those points which have been determined to be sufficiently deep that they are with high probability in the top-$k$.
For simplicial median computation this decision is one-sided.
For coarse-ranking this elimination procedure is necessarily slightly more sophisticated as discussed by \citet{karpov2020batched}.

\begin{algorithm}[h]
  \caption{\texttt{Adaptive Data Depth Meta Algorithm} \label{alg:adaDepthMeta}}
\begin{algorithmic}[1] 
    \State \textbf{Input: } data set $\CD = x_1,\hdots,x_n\in\R^d$, error probability $\delta$, exact computation method \edepth\ with runtime $\ecost$, approximation method \adepth\ with runtime $\acost(\epsilon,\delta)$
    \State \textbf{Initialize: } $S_1 \gets [n]$, $r=0$
    \Repeat
    \State $r\gets r+1$
    \State Set $\epsilon_r = 2^{-r}$
    \If{$\acost(\epsilon_r/2,\delta/(2nr^2)) > \ecost$} \Comment{Use specialized exact method} \label{line:exactCompMeta}
    \State Exactly compute $\hat{\mu}_i^r \gets \edepth(x_i,\CD)$ for $i \in S_r$
    \State \textbf{Break}
    \EndIf
    \State For $i \in S_r$ compute $\hat{\mu}_i^r\gets \adepth(x_i,\CD,\epsilon_r/2,\delta/(2nr^2))$
    \State Construct $E_r\subseteq S_r$ as set of arms to eliminate, based on $\{\hat{\mu}_i^r\}$
    \State Set $S_{r+1}\gets S_r \setminus E_r $
  \Until \texttt{terminationCondition}($\{\hat{\mu}_i^r\},\epsilon,\epsilon_r$)
  \State \Return Solution as a function of $S_r,\{\hat{\mu}_i^r\}$
\end{algorithmic}
\end{algorithm}

Extending the analysis of \Cref{alg:adaDepth} to \Cref{alg:adaDepthMeta}, we are able to provide the following Theorem regarding the performance of our meta algorithm when specified to simplicial median computation, using  general methods for exact and approximate computation \edepth\ and \adepth\ with runtimes $\ecost$ and $\acost(\epsilon,\delta)$ respectively.

\clearpage
\begin{thm}[Meta Algorithm]\label{thm:metaAlg}
Algorithm \ref{alg:adaDepthMeta} succeeds with probability at least $1-\delta$ in returning a point $x_i\in \CD$ with simplicial depth within an additive $\epsilon$ of the simplicial median's depth, requiring computation at most
\begin{equation*}
\sum_{i=1}^n \left[\left( \sum_{r=1}^{r_i-1}\acost \left(2^{-r-1},\delta/(2nr^2)\right) \right)+\min\left(\acost\left(\frac{\Delta_i}{8},\frac{\delta}{2nr_i^2}\right), \ecost\right)\right],
\end{equation*}
where %
$r_i=\min \left(\{r : r\in\N, \acost(2^{-r-1},\delta/(2nr^2)) \ge \ecost\} \cup \{\lceil \log (2/\epsilon)\rceil,\lceil \log_2 (2/\Delta_i) \rceil\}\right)$.
\end{thm}

To simplify this expression, we make the assumption that $2\times\acost(\epsilon,\delta) \le \acost(\epsilon/2,\delta)$ for all $\epsilon\le 1$, and that this function is increasing as $\delta$ decreases.
Note that this holds for our sampling based approximation scheme, with $4\times\acost(\epsilon,\delta)= \acost(\epsilon/2,\delta)$.
\begin{cor}\label{cor:metaAlg}
Algorithm \ref{alg:adaDepthMeta} succeeds with probability at least $1-\delta$ in returning the simplicial median of the data set, and assuming that $\acost(\epsilon/c,\delta)$ scales superlinearly in $c$ requires time at most
\begin{equation*}
\sum_{i=1}^n \left[ \min\left(2\times\acost\left(\frac{\Delta_i}{8},\frac{\delta}{2nr_i^2}\right), 3\times \ecost\right)\right],
\end{equation*}
where $r_i=\min \left(\{r : r\in\N, \acost(2^{-r-1},\delta/(2nr^2)) \ge \ecost\} \cup \{\lceil \log (2/\epsilon)\rceil,\lceil \log_2 (2/\Delta_i) \rceil\}\right)$.
\end{cor}

This modularity allows for the incorporation of other approximation algorithms.
For example, the high-dimensional approximation scheme proposed by \citet{afshani2015approximating}  could be used (which works for relative approximations and so our meta-algorithm would need to be correspondingly modified).
The runtime of this scheme is $\tilde{O}(n^{d/2+1})$, and so one potential scheme could be to use random sampling of simplices to coarsely approximate the simplicial depth of points in early rounds, the more sophisticated geometric approach for approximating the depth of a point to a higher accuracy by \citet{afshani2015approximating} in later rounds for small $\epsilon$, and a specialized exact computation method for those points with near maximal simpicial depth.

\vspace{-.1cm}
\section{Numerical Simulations} \label{sec:numerics}
We simulate our adaptive simplicial median computation scheme on synthetic data sets and show the empirical performance improvement afforded by adaptivity.
By adaptively approximating the simplicial depth of points in our data set to the necessary accuracy, we are able to obtain significant computational improvements.
All experimental results can be reproduced from our publicly available code: \url{https://github.com/TavorB/adaSimplicialDepth}.

\begin{figure}[h]
\vspace{-.6cm}
    \centering
    \begin{subfigure}[b]{0.45\textwidth}
    \includegraphics[width=\textwidth]{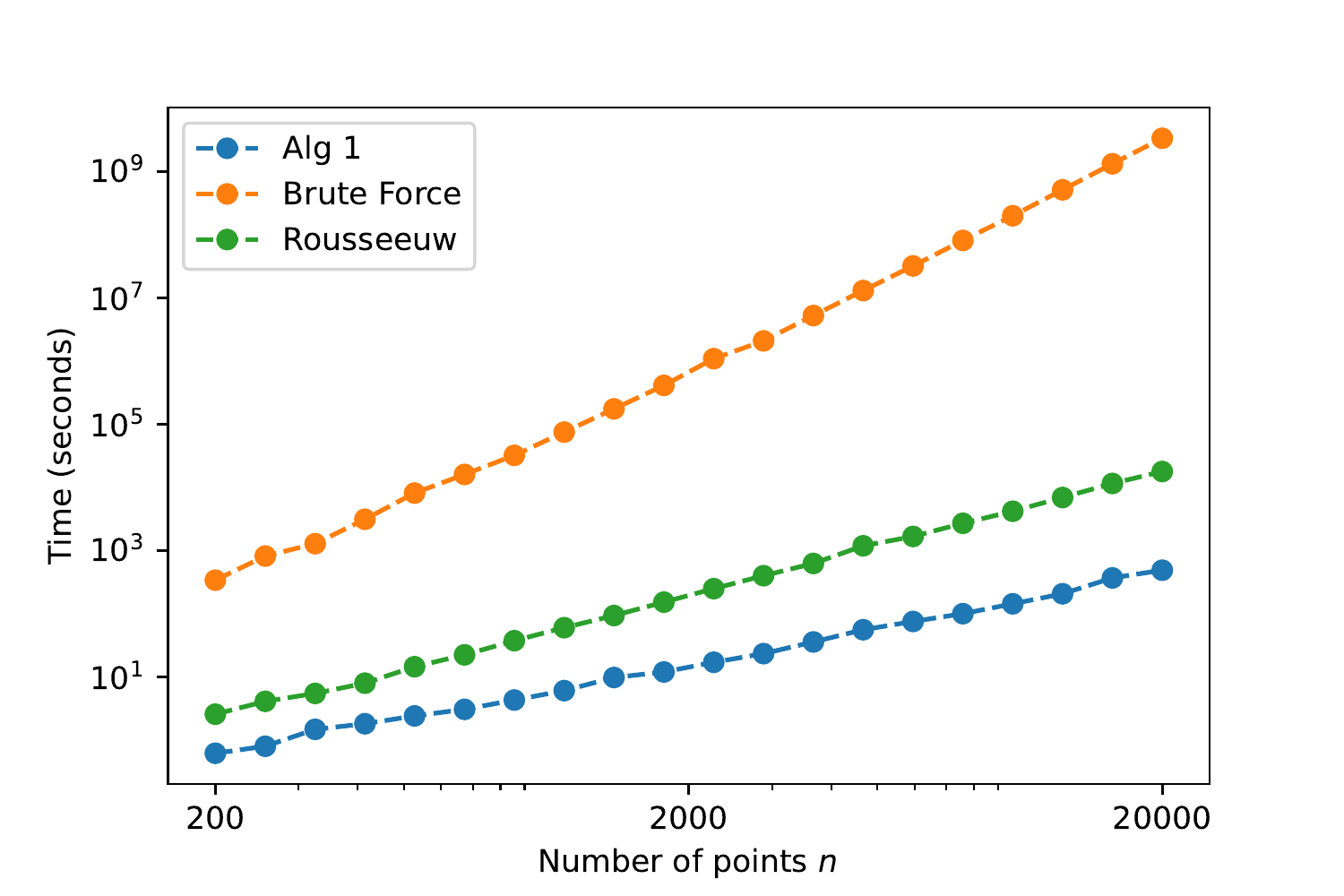}
    \subcaption{}
    \label{fig:2dBruteRutsAda}
    \end{subfigure}
    \begin{subfigure}[b]{0.45\textwidth}
    \includegraphics[width=\textwidth]{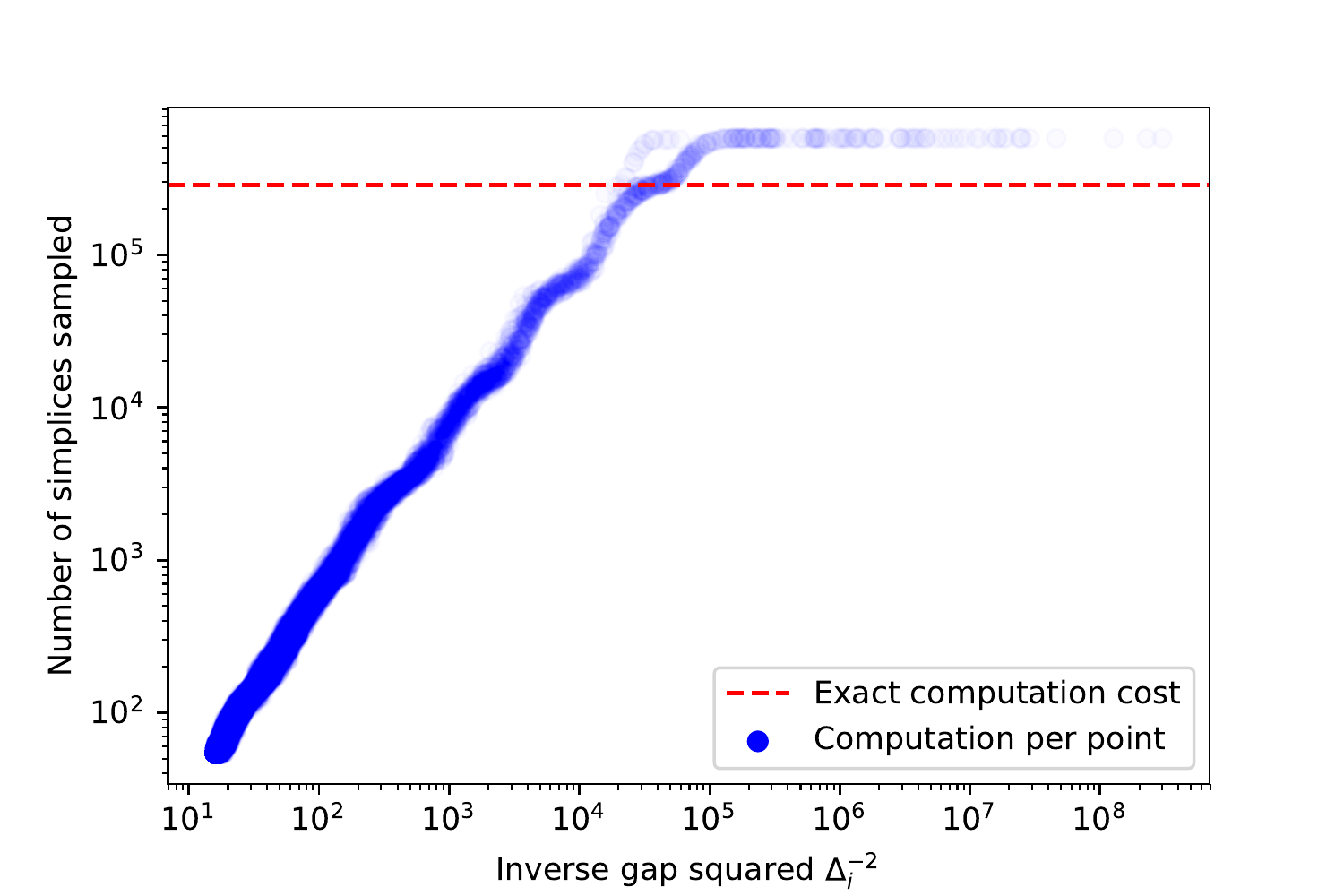}
    \subcaption{}
    \label{fig:2dnumPullsvsGap}
    \end{subfigure}
    \vspace{-.3cm}
    \caption{Simulations for simplicial median computation for $d=2$, data drawn from isotropic gaussian.
    a) shows wall clock time with 20 simulations per point; 10 random data sets, 2 trials per data set, further details in \Cref{app:expDetails}. Adaptive method run with $\delta= .001$ returns the correct answer on every run.
    b) shows the number of computations made by our adaptive algorithm on a given point as a function of inverse gap squared (in terms of simplicial depth) for $n=20,000$, averaged over 50 trials on the same randomly generated data set. Plotted in log-log scale to show points with small gaps. $R^2= .98$ between inverse gap squared $x$ and number of simplices sampled $y$ for points with $y<\ecost$, validating that this relationship is indeed nearly linear.}
    \vspace{-.6cm}
\end{figure}

As can be seen in \Cref{fig:2dBruteRutsAda}, our adaptive algorithm dramatically outperforms both exact computation and the geometric but nonadaptive solution of Rousseeuw and Ruts, yielding superior scaling with $n$.
Plotting the runtime of these different methods in log-log scale, we see that the brute force algorithm has a wall clock time scaling theoretically as $O(n^4)$, but empirically as roughly $O(n^{3.4})$, due to the efficiency of batch operations in computation of whether $n$ points are contained within a simplex, which practically scales sublinearly in $n$.
Rousseeuw's algorithm has a runtime which theoretically scales as $O(n^2)$, and is practically validated as such with an empirical slope of $2.0$.
Our algorithm has an instance dependent runtime, where for isotropic Gaussians the runtime scales approximately as $O(n^{1.5})$.
Full experimental details are laid out for reproducibility in \Cref{app:expDetails}.

We can see that the number of pulls required per arm scales as expected in \Cref{fig:2dnumPullsvsGap}. Plotting the number of computations required for point $x_i$ (which we denote as $T_i$) as a function of $\Delta_i^{-2}$, we see a linear relationship between $T_i$ and $\Delta_i^{-2}$ up until $T_i$ exceeds the threshold for exact computation, at which point the number of pulls required is constant.

Examining the reason for the dramatic gain in \Cref{fig:2dBruteRutsAda}, we see in \Cref{fig:2dGapHist} that only a very small fraction of points require exact computation of their simplicial depth.
This highlights the gain of our adaptive method; many points require only a very coarse approximation of their depth before they can be eliminated.
In these simulations, no errors were recorded for our adaptive scheme.
In all trials the threshold for deciding to compute a point's depth exactly was if more than $n\log n$ simplicies needed to be computed for it, i.e. $t_r \ge n\log n$. 
This, in concert with the fact that only approximate gaps based on the algorithm's output were plotted, lead to a jagged end behavior in \Cref{fig:2dGapHist}.

\begin{figure}[h]
    \vspace{-.2cm}
    \centering
    \begin{subfigure}[b]{0.45\textwidth}
    \includegraphics[width=\textwidth]{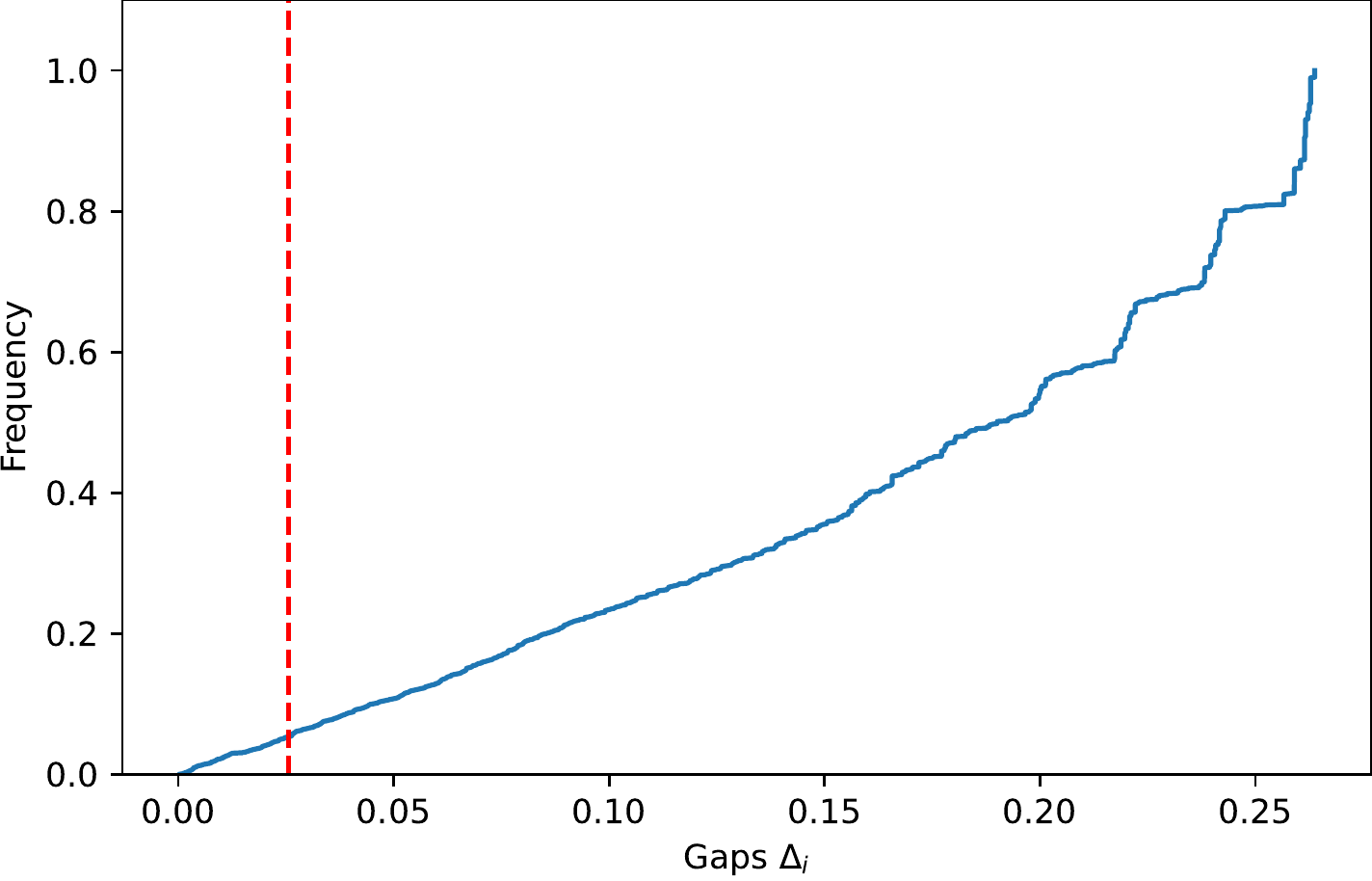}
    \subcaption{}
    \end{subfigure}
    \quad
    \begin{subfigure}[b]{0.45\textwidth}
    \includegraphics[width=\textwidth]{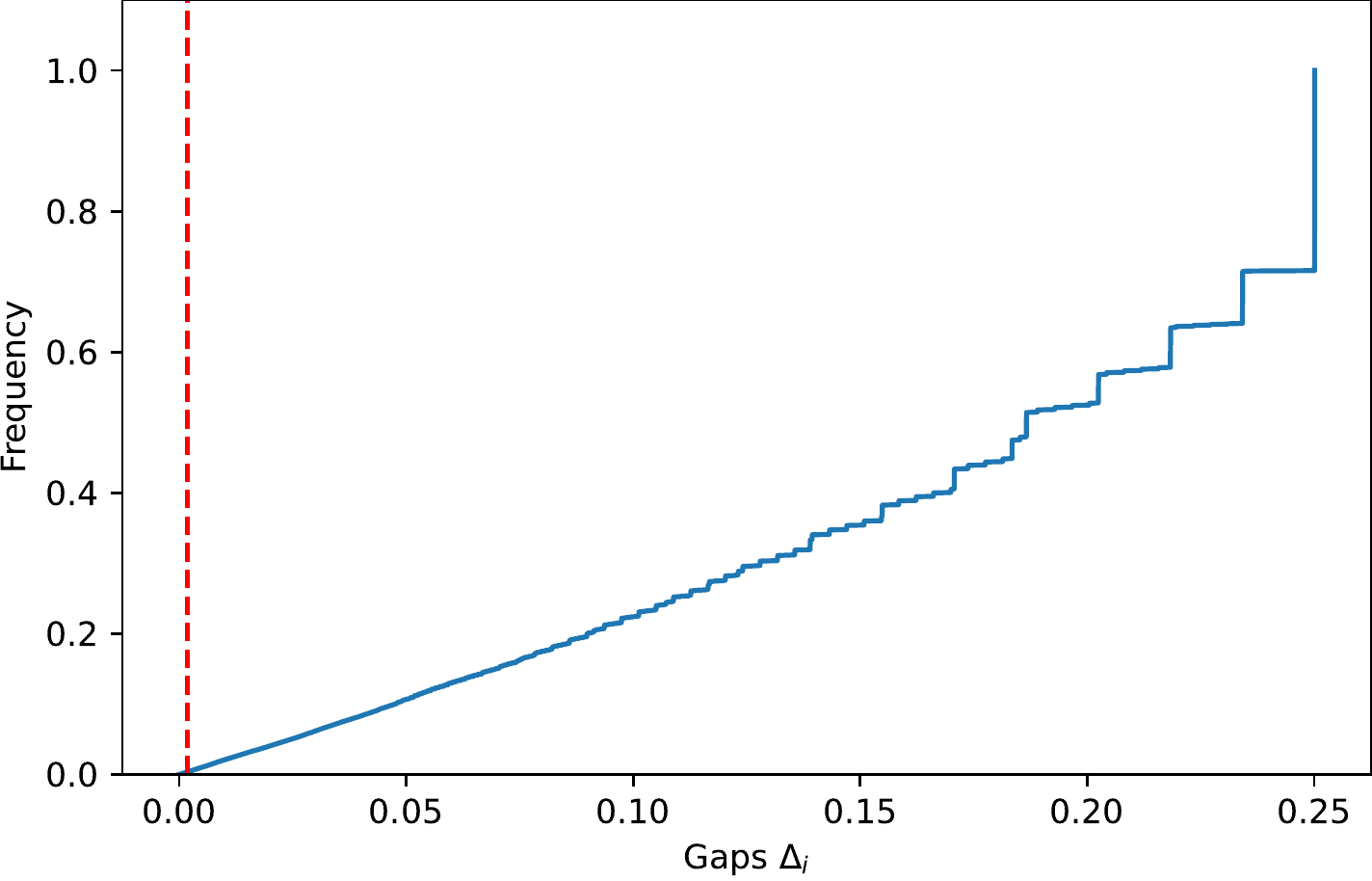}
    \subcaption{}
    \end{subfigure}
    \vspace{-.3cm}
    \caption{CDF of gaps when data is drawn from an isotropic Gaussian in $d=2$ for a) $n=200$, b) $n=20k$. Red line indicates exact computation threshold. Averaged over 10 random instances, with 2 trials per instance, gaps approximately computed using our adaptive method. Fitting (a) to a power law distribution yields $\alpha=1.15$ and (b) similarly yields $\alpha=1.16$, showing the dramatic improvement of our method, where our gain increases as $\alpha$ gets larger.}\label{fig:2dGapHist}
\end{figure}

We see in \Cref{fig:2dGapHist} the gap distribution for $n=200$ and $n=20k$, and observe that they are quite similar.
This trend held for all $n$ tested, and so only the endpoints were shown; the rest of the plots can be found in \Cref{fig:2dGapHistApp} in the Appendix.
Observe that the maximum gap stays roughly constant at $D_n(x_i)\approx .25$.
This is because for rotationally symmetric distributions, like the isotropic Gaussian we utilize, the probability that the origin is contained within the triangle constructed from 3 randomly drawn samples from this distribution is $1/4$.
Since many samples are drawn from this 2-dimensional distribution, the simplicial median is close to the origin, and has a simplicial depth of approximately $1/4$, compared with points on the outside of the distribution which have a simplicial depth close to 0.
One natural fit for the empirical CDF in \Cref{fig:2dGapHist} is a power law distribution, as discussed in \Cref{cor:powerLaw}.

\section{Discussion} \label{sec:disc}
With our base algorithm and its extensions in place, we discuss two additional points regarding algorithmic design.
\subsection{Adapting to Variances is Unnecessary}
Viewing this problem from the multi-armed bandit perspective, we see that since our arm rewards are Bernoulli, our arm pulls are bounded leading to an immediate bound on the sub-Gaussian parameter of the arm distributions.
As we will show, this coarse bound is sufficient; adapting to the variance of the arms can yield only a constant factor improvement for simplicial median identification when considering the refined definition of simplicial depth as the average number of open and closed simplices that contain the query point \citep{burr2006simplicial}.

For a data set of $n$ points in $\R^2$, \citet{boros1984number} showed that there exists a point contained within at least $n^3/27 + O(n^2)$ simplicies (triangles), i.e. a simplicial depth of $\Theta(1)$.
For higher dimensions, the same authors show that a point can be contained in no more than a $2^{-d}+ o(1)$ fraction of simplicies, translating to a simplicial depth of at most $2^{-d} + o(1)$, which is $\Theta(1)$ for fixed $d$ as $n$ grows.
Conversely, it was shown by \citet{barany1982generalization} and discussed in the context of simplicial depth by \citet{gil1992geometric} that there must exist a point contained within at least a $\frac{1}{(d+1)^{d+1}}$ fraction of simplices, which is $\Theta(1)$ for constant $d$.
Taking these together we have that $\mu_1$, the depth of the simplicial median,  is bounded away from both 0 and 1 by constants, which are functions of $d$ but independent of $n$ as $n\to \infty$.

At first glance, this may appear worrisome; if certain points have very small simplicial depth on the order of $1/n$ then simply using a Hoeffding-type inequality relying on the boundedness of the random variables in question yields slow concentration, as it ignores the smaller variance of these arm, which is $\Theta(\frac{1}{n})$ not $\Theta(1)$.
We then would like confidence intervals whose width depends on the variance of the arm in question, for example those offered by empirical variants of Bernstein's inequality as developed by \citet{maurer2009empirical}. %
Several recent works have focused on pure exploration settings with unknown variances \citep{lu2021variance}, which yield results applicable to this setting of arms with unknown means and variances but bounded rewards. 
As was shown by \citet{lu2021variance}, the optimal sample complexity in the case where arms have unknown means $\mu_i$ and variances $\sigma_i^2$ is essentially lower bounded by
\begin{equation} \label{eq:varLB}
    \tilde{\Omega}\left(\sum_{i=1}^n \left[ \frac{\sigma_i^2}{\Delta_i^2} + \frac{1}{\Delta_i} \right]\log\left(\frac{1}{\delta}\right)\right).
\end{equation}
Since $\sigma_i^2 = \mu_i (1-\mu_i)$ for these Bernoulli rewards, we can simplify \eqref{eq:varLB} in our setting to bound the cost of eliminating the $i$-th arm as
\begin{equation}
    \frac{\sigma_i^2}{\Delta_i^2} +\frac{1}{\Delta_i}
    = \frac{\mu_1-\mu_i^2}{\Delta_i^2} 
    \ge \frac{\mu_1-\mu_1^2}{\Delta_i^2}
    = \Theta(\Delta_i^{-2}),
\end{equation}
where the final equality follows from the fact that $\mu_1$ is upper and lower bounded by constants with respect to $n$.
This shows that variance adaptation can only yield a $d$ dependent constant factor improvement in this setting, and thus that adapting to the variance of the arms is unnecessary to achieve order optimal performance.
The intuition behind this is that if the gap $\Delta_i$ is small then the arm mean $\mu_i$ must be large, $\Theta(1)$, and so its variance is constant.
On the other hand if the gap is large and $\Delta_i=\Theta(1)$,then $\Delta_i^{-1}$ and $\Delta_i^{-2}$ are the same up to constants, and so this arm does not require many samples to eliminate.
For alternative objectives which necessitate differentiating between points with shallow simplicial depth, variance adaptation may yield more substantial improvements.
In these instances, techniques similar to those used by \citet{lu2021variance} can be readily employed in our algorithms for alternative objectives, as discussed in the two subsequent sections.

\subsection{Barycentric Coordinate Computation} \label{sec:barycentric}
One aspect of simplicial depth computation that we did not focus on in this work is how arm pulls are actually performed.
That is, for a given point $x_0 \in \R^d$, how do we check whether it is in the simplex created by $d+1$ other points?
Verifying if $x_0$ is contained within the convex hull (simplex) of $d+1$ points $x_1,\hdots,x_{d+1}$ is equivalent to computing the barycentric coordinates of $x_0$ with respect to $\{x_i\}_{i=1}^{d+1}$, and verifying that all are positive.
Mathematically, denoting the normalized barycentric coordinates of $x_0$ as $\lambda \in\R^{d+1}$ where $\sum_i \lambda_i = 1$, we wish to express
$x_0 = \sum_{i=1}^{d+1} \lambda_i x_i.$ 
To solve for such a $\lambda$, we define our centered data matrix $X=[x_1-x_{d+1}\  x_2-x_{d+1}\  \hdots\  x_d-x_{d+1}]$, and see that defining $\lambda_{[d]}$ as the first $d$ coordinates of $\lambda$, we have $X \lambda_{[d]} = x_0-x_{d+1}$ and $\lambda_{d+1}=1-\sum_{i=1}^d \lambda_i$.
$x_0$ is contained within the closed simplex with vertices $\{x_i\}_{i=1}^{d+1}$ if and only if $\lambda_i\ge0$ for all $i$.
Note that while solving this for one query point $x_0$ requires 1 linear system solve, practically requiring $O(d^3)$ time, solving this for several different query points simultaneously is much more efficient.
If we precompute $X^{-1}$ (or an LU decomposition of $X$), then computing the barycentric coordinates for an additional point $x_0'$ with respect to the same set of points $\{x_i\}_{i=1}^{d+1}$ will require only a matrix vector multiplication (or backsolving), requiring $O(d^2)$ time.
This means that while obtaining $k$ samples for the simplicial depth of a single point $x_0$ requires $O(kd^3)$ time, obtaining 1 sample each for $k$ different points requires only $O(d^3 + kd^2)$ time, a factor of $d$ improvement.
More importantly, solving many systems of linear equations for the same data matrix $X$ is efficient from a computational perspective, as it is amenable to being pipelined at a BLAS level.
To this end, our algorithm is optimized for batch efficiency, operating in a number of rounds scaling logarithmically with $1/\epsilon$ or $1/\ecost$, which is the minimal possible to achieve order optimal sample complexity up to logarithmic factors, as discussed by \citet{karpov2020batched}. Within each round, our uniform sampling of arms allows us to obtain this theoretical factor of $d$ improvement as well as the practical benefits of more efficient batched operations.

\section{Concluding Remarks} \label{sec:conc}

In this work we proposed a novel algorithmic framework for the adaptive computation of data depth.
This method enables us to compute the simplicial median of a data set efficiently by approximating the depth of each point to the necessary accuracy, with dramatic time savings allowing for these computations to be run on larger and higher dimensional data sets than was previously possible.
We provided instance dependent theoretical guarantees for our adaptive method, and showed its excellent empirical performance.

In addition to simplicial depth, our proposed technique of using adaptivity to efficiently approximate the points' depths to the necessary accuracy can be extended to several other common measures of depth including majority depth, Oja depth, and likelihood depth \citep{liu1999multivariate}.
Majority depth is defined with respect to half-spaces, where the sample majority depth is defined as $M_n(x)= \E\left[ \mathds{1}\{x \text{ is in a major side determined by } (X_1,\hdots,X_d)\}\right]$, where major side indicates the half-space bounded by the hyperplane containing $(X_1,\hdots,X_d)$ which has probability at least one half, where $(X_1,\hdots,X_d)$ are drawn uniformly at random  from all subsets of $\{x_i\}$ of size $d$.
This is naturally amenable to our technique of adaptive sampling, as now instead of computing whether $x$ is within a simplex of $d+1$ random points, we rather check whether it is on a given side of a hyperplane defined by $d$ points, and adaptively approximate the more promising points.
The sample Oja depth is defined as $OD_n(x) := \left(1+\E\left[\text{volume}(S[x,X_1,\hdots,X_d\right]\right)^{-1}$, for $X_1,\hdots,X_d$ drawn uniformly at random from all subsets of $\{x_i\}$ of size $d$, where $S$ is the closed simplex formed by the $d+1$ points.
Adaptive sampling can be incorporated using our techniques, as not all $n \choose d$ subsets need to be enumerated for each point in order to determine whether it is the deepest or not.
Note that we are no longer obtaining unbiased estimates of the Oja depth itself, and now essentially need to transfer the confidence intervals from $y$ to $(1+y)^{-1}$ where $y= \E\left[\text{volume}(S[x,X_1,\hdots,X_d\right]$ is the quantity we approximate.
A final common measure of depth that we highlight is likelihood depth, which can also be adaptively approximated in certain cases.
The distributional likelihood depth for a point $x$ drawn from a distribution with density $f$ is $L(x)=f(x)$, where the sample version is based on an empirical estimate of the density at $x$.
For kernel-based density estimates adaptivity can be utilized, as then $L_n(x) = \frac{1}{n}\sum_{i=1}^n k(x,x_i)$, at which point we observe that not all $n$ kernel computations need to be performed for each point.
These alternative notions of depth show the natural transferability of our algorithmic principle; if \textit{relative} ordering is the object of interest, adaptivity can yield dramatic gains.

There are several important lines of future work regarding adaptive simplicial depth computation.
One is to theoretically analyze the effects of arm correlation on the performance of Algorithm 1.
In the work of  \citet{baharav2019ultra} it was shown that correlated arm pulls could provably improve sample complexity in the fixed budget setting.
In this fixed confidence regime however, a more sophisticated scheme is required to estimate and exploit these unknown dependencies.
Another direction is to note that there exists spatial information that we are neglecting; 
if two points $x_1,x_2$ are close to each other then $D_n(x_1)$ should not be too far from $D_n(x_2)$, depending on the positions of the other points.
This can potentially improve the algorithmic dependence on $n$, necessitating a more sophisticated bandit algorithm based on bandits in metric spaces \citep{mason2019learning}.

\acks{Baharav was supported in part by the NSF GRFP and the Alcatel-Lucent Stanford Graduate Fellowship.
Lai was supported in part by the NSF under DMS-1811818.
}

\newpage

\appendix
\section{Proofs} \label{app:additionalProofs}
In this Appendix, we provide the full proofs missing from the main text.

\subsection{Proof of \Cref{thm:main}}
We analyze \Cref{alg:adaDepth} following the steps of \citet{hillel2013distributed}, with the necessary modifications for the BMO method similar to those made by \citet{bagaria2021bandit}.
We begin by showing that each depth estimator $\hat{\mu}_i^{r}$ is within $\epsilon_r/2$ of the true simplicial depth of point $x_i$, $\mu_i$, for all $i,r$ with probability at least $1-\delta$.
We do this by assuming that samples are drawn for all arms in each round, even if they are not observed, and bound the probability that any of these empirical estimates deviate from their mean.

\begin{lem} \label{lem:conf}
With probability at least $1-\delta$ we have that $|\hat{\mu}_i^{r} - \mu_i| < \epsilon_r/2$
for all $i \in [n]$ and all $r \in \N$ simultaneously, where
\begin{equation*}
    \hat{\mu}_i^{r} := \frac{1}{t_r} \sum_{j=1}^{t_r} \mathds{1}\{x_i\in S[x_{i_1}^{(j)},\hdots,x_{i_{d+1}}^{(j)}]\}
\end{equation*}
and $\{x_{i_1}^{(j)},\hdots,x_{i_{d+1}}^{(j)} \}$ is drawn uniformly at random without replacement from $\CD$, independently for each $j\in [t_r]$.
\end{lem}
\begin{proof}[Proof of \Cref{lem:conf}] \belowdisplayskip=-12pt
By Hoeffding's inequality for $[0,1]$ bounded random variables \citep{wainwright2019high}, we have that
\begin{align*}
    \P \left( \bigcup_{i,r} \{|\hat{\mu}_i^{r} - \mu_i| \ge \epsilon_r/2\} \right) 
    &\le \sum_{i,r}  \P \left( |\hat{\mu}_i^{r} - \mu_i| \ge \epsilon_r/2 \right) \\
    &\le \sum_{i,r} 2\exp\left(-2t_r (\epsilon_r/2)^2\right)\\ 
    &\le \sum_{r=1}^\infty \frac{\delta}{2r^2}\\
    &\le \delta %
\end{align*}
\end{proof}

\noindent
With this lemma in place, we are now able to analyze \Cref{alg:adaDepth}.

\begin{proof}
Conditioning on the event in \Cref{lem:conf} where our confidence intervals hold, we have that the simplicial median (best arm) is not eliminated during the course of the algorithm.
Further, we see that any suboptimal arm $i$ must be eliminated by round $r_i = \lceil \log_2 (2/\Delta_i)\rceil$, since then $\Delta_i \ge 2\epsilon_r$, and so $\hat{\mu}_{*}^r-\epsilon_r > \mu_1 - 3\epsilon_r/2 \ge \mu_i + \epsilon_r/2 > \hat{\mu}_i^r$.
Thus, the algorithm must terminate after at most $\lceil \log_2 (2/\Delta_2)\rceil$ rounds, at which point only the simplicial median will remain.
If multiple points attain the maximum simplicial depth, then the algorithm will still terminate within at most $\lceil \log_2 (2\sqrt{\ecost})\rceil$ rounds due to the doubling accuracy.

Analyzing the algorithm's requisite sample complexity, we see that since arm $i$ must be eliminated prior to round $r_i=\lceil \log_2 (2/\Delta_i) \rceil$, it can be pulled no more than
\begin{align*}
    t_{r_i}&= \lceil2 \epsilon_{r_i}^{-2} \log (4nr_i^2/\delta) \rceil\\
    &\le 32\Delta_i^{-2}\log (16n\log_2^2 (2/\Delta_i)/\delta)  + 1 \numberthis
\end{align*}
times. Here, we used that $r_i \ge 1$ and so $r_i\le 2\log_2 (2/\Delta_i)$.
Additionally, due to our ability to exactly compute the simplicial depth of a point, our algorithm will not pull any single arm too many times.
Concretely, since the mean of arm $i$ is exactly computed if $t_{r_i}\ge \ecost$, at most $\ecost$ computation is done before $x_i$'s depth is exactly computed, i.e. at most $2\times\ecost$ total work.

Since the algorithm must terminate if there is only one remaining arm, we have that the best arm will be pulled the same number of times as the second most pulled arm, and so defining $r_1:= r_2$ and $\Delta_1:=\Delta_2$ we have that on the event where the confidence intervals hold, the total sample complexity will be at most
\begin{align*}
   \sum_{i=1}^n t_{r_i} 
   &\le n + \sum_{i=1}^n \min\left(\frac{32 \log \left(\frac{16n}{\delta}\log_2^2 \left(\frac{2}{\Delta_i}\right)\right)}{\Delta_i^2}, 2\times \ecost\right). \numberthis \quad %
\end{align*}
\end{proof}

\vspace{-1cm}
\subsection{Proof of \Cref{cor:powerLaw}} \label{app:powerLaw}
In this subsection we show that when the gaps in the data follow a power law distribution, the sample complexity of \Cref{alg:adaDepth} is significantly better than the naive $O(n^d)$ required for simplicial median identification.
This average case analysis where the gaps follow a power law distribution follows similarly to that of Corollary 1 in \citep{bagaria2021bandit}.

\begin{proof}[Proof of \Cref{cor:powerLaw}]
We have by Theorem \ref{thm:main} that with probability at least $1-\delta$ Algorithm \ref{alg:adaDepth} will successfully return the simplicial median.
Since the event where our confidence intervals hold is independent of the random gaps that are drawn, we can integrate out the expected sample complexity with respect to the random gaps satisfying $F(\Delta)=\Delta^\alpha$.
Thus, on this success event our number of samples $M$ satisfies
\begin{align*}
    \E\{M\} 
    &= O\left( \E\left\{\sum_{i=1}^n \min\left( \frac{\log\left(\frac{n}{\delta}\log(\frac{1}{\Delta_i})\right)}{\Delta_i^2}, \ecost\right) \right\}\right)\\
    &=O\left(n\int_{\Delta=0}^1 \min\left( \frac{\log\left(\frac{n}{\delta}\log(\frac{1}{\Delta})\right)}{\Delta^2}, \ecost \right)f(d\Delta) \right)\\
    &= O\left(  n \left(\ecost^{1-\alpha/2}+ \log\left(\frac{nd \log n}{\delta}\right) \int_{\Delta=\ecost^{-1/2}}^1  \hspace{-1.2cm}\alpha \Delta^{\alpha-3} d\Delta \right)\right)\\
    &= \begin{cases} 
      O\left(n\log \left(nd/\delta\right) \ecost^{1-\alpha/2}\right), & \textnormal{for }\alpha \in [0,2), \\
      O\left(n \log \left(nd/\delta\right) \log (\ecost) \right), & \textnormal{for } \alpha = 2, \\
      O\left(n \log \left(nd/\delta\right) \right), & \textnormal{for } \alpha > 2. 
   \end{cases}
   \numberthis\label{eq:PropEq}
\end{align*}
noting that $\log(1/\Delta) \le d\log n$ for $\Delta \ge \ecost^{-1/2}$.
\end{proof}

\subsection{Proof of \Cref{thm:simplicialTopk}} \label{app:topk}
We analyze \Cref{alg:adaDepthAlpha} similarly to \Cref{alg:adaDepth}, showing that the good event $\xi$ where our mean estimators stay within their confidence intervals occurs with high probability. Then, on this good event, our algorithm will correctly identify the $k$-deepest elements, up to an allowable additive $\epsilon$.
\begin{proof}[Proof of \Cref{thm:simplicialTopk}]
As in the proof of \Cref{thm:main} we assume that all arms are sampled in all rounds until termination, even if they are not observed, and bound the probability that any of these estimators $\{\hat{\mu}_i^r\}_{i,r}$ stray outside of their confidence intervals via \Cref{lem:conf}. 
Denote this good event in Lemma 1 as $\xi$.

Prior to the termination condition in \cref{line:topKIf}, we see that on the event $\xi$ we have that 1) no point $i$ with $\mu_i<\mu_k$ will be added to the acceptance set $A$, 
2) no point $i$ with $\mu_i>\mu_k$ will be eliminated without being added to $A$, 
3) all points $i$ with $\Delta_i \ge 2\epsilon_r$ will be removed from the active set by the end of round $r$.
These can be seen by induction on $r$, on the good event $\xi$.
With these observations in hand, we see that when the algorithm exits the sampling loop, all arms with $\Delta_i \ge \epsilon$ will have been removed from the active set.
Examining these points by cases, we see that for those points with $\mu_i > \mu_k + \epsilon$, they will necessarily be contained within the output set $A$.
Conversely, those points with $\mu_i < \mu_k - \epsilon$ will necessarily not be contained within the output set $A$.
Thus, we have that on this event $\xi$ the set $A$ satisfies
\begin{equation}
    \{i : \mu_i > \mu_{k}+\epsilon\} \subseteq A \subseteq \{i : \mu_i \ge \mu_{k} - \epsilon\}.
\end{equation}

Bounding the number of pulls required, we see that a point $i$ must be eliminated by round $r_i=\lceil \log_2(2/\Delta_i^{(k)})\rceil$.
Either it will have been eliminated by round $r_i$ where $\Delta_i^{(k)} \ge 2\epsilon_{r_i}$, or we will have $\epsilon_r \le \epsilon/2$ in which case the algorithm will terminate.
Thus, our effective $\Delta_i^{(k)}$ satisfies $\Delta_i^{(k)} = \max(\mu_i - \mu_{k+1},\epsilon)$ if $i\le k$ and $\Delta_i^{(k)} = \max(\mu_k - \mu_{i},\epsilon)$ if $i> k$, as estimating the mean of point $i$ to accuracy $\epsilon/2$ is also sufficient.
Following the argument in \Cref{thm:main}, this implies that the number of pulls required to eliminate arm $i$ satisfies
\begin{align*}
    t_{r_i}&= \lceil2 \epsilon_{r_i}^{-2} \log (4nr_i^2/\delta) \rceil\\
    &\le 32\left(\Delta_i^{(k)}\right)^{-2}\log (16n\log_2^2 (2/\Delta_i^{(k)})/\delta)  + 1. \numberthis
\end{align*}
Utilizing the fact that our algorithm also terminates if the approximation cost is deemed to be greater than the exact computation cost, we have a total sample complexity on $\xi$ of
\begin{equation}
    O\left( \sum_{i=1}^n \min\left(\frac{\log \left( \frac{n}{\delta} \log \left(1/\Delta_i^{(k)}\right)\right)}{\left(\Delta_i^{(k)}\right)^2}, \ecost \right)\right),
\end{equation}
giving us the desired result.
\end{proof}

\subsection{Proof of \Cref{thm:metaAlg}}

In this section we provide the proof for our meta algorithm when applied to approximate simplicial depth computation.

We see by the careful choice of failure probabilities in line 10 that all our $\adepth$ calls will be correct with probability at least $1-\delta$.
This is because our failure probability is at most
\begin{align*}
    \P(\text{failure}) & = \P\left( \bigcup_{r \in \N} \bigcup_{i \in S_r} \{\adepth \text{ confidence intervals hold}\} \right)\\
    &\le \sum_{r \in \N} \sum_{i \in S_r} \P\left(\adepth \text{ confidence intervals hold}\right)\\
    &\le \sum_{r \in \N} \frac{\delta}{2r^2}\\
    &\le \delta.
\end{align*}

Conditioning on the good event $\xi$ where our \adepth\ confidence intervals hold, we see that in the case of approximate simplicial median identification we have that arm $i$ must be eliminated by the end of round $r_i$, where
\begin{equation}
    r_i=\min \left(\{r : r\in\N, \acost(2^{-r-1},\delta/(2nr^2)) \ge \ecost\} \cup \{\lceil \log (2/\epsilon)\rceil,\lceil \log_2 (2/\Delta_i) \rceil\}\right).
\end{equation}
We see this by contradiction; first, assuming that arm $i$ is active in round $r$ where $r> r_i$.
This implies that either $r>\lceil \log_2 (2/\Delta_i) \rceil$, which we know cannot happen on the good event $\xi$ where our confidence intervals hold, as was shown in the proof of \Cref{thm:main}.
If $r>\lceil \log_2 (2/\epsilon) \rceil$, then we must have that $\epsilon_r <\epsilon/2$, which is included as a termination condition for our meta algorithm in this setting, as if all points have depths estimated to accuracy $\epsilon/2$, then the point with the largest estimated simplicial depth must be within $\epsilon$ of the maximum.
Finally, if $r > \argmin_{r \in \N} \acost(2^{-r-1},\delta/(2nr^2)) \ge \ecost$, then the depth of point $i$ would be exactly computed.

Thus, this algorithm (when applied to approximate simplicial median computation) will spend at most $r_i$ rounds approximating the depth of point $i$, and in total spend at most 
\begin{equation}
    \left( \sum_{r=1}^{r_i-1}\acost \left(2^{-r-1},\delta/(2nr^2)\right) \right)+\min\left(\acost\left(\frac{\Delta_i}{8},\frac{\delta}{2nr_i^2}\right), \ecost\right)
\end{equation}
computation on point $i$.
Summing over the points in the data set yields the desired result.

\section{Experimental Details}\label{app:expDetails}

Here we provide details for reproducing the simulation results shown in the paper.
All our code is publicly available on GitHub for reproducibility: \url{https://github.com/TavorB/adaSimplicialDepth}.

The data for each experiment was generated as $n$ independent samples from an isotropic gaussian.
We use the original definition of simplicial depth, as the probability that a random \textit{closed} simplex contains the query point.
All experiments were run on one core of an AMD Opteron Processor 6378 with 500GB memory (no parallelism within a trial). 
All adaptive experiments are run with 20 trials per point.
As simple Hoeffding-based confidence intervals are known to be overly conservative for multi-armed bandits in practice, to obtain tighter confidence intervals we set $t_r =\lceil .2 \times 2^{-r} \log (4nr^2/\delta)\rceil$ in all our simulations.
The threshold for exact computation was determined as whether $t_r \ge n \log n$.
Trials are grouped in pairs, where the same data set is used for both trials, but a different random seed is used to facilitate a different random selection of simplicies.
Since computing ground truth results is computationally infeasible for all but the smallest of data sets, we declare the run a success and the correct result returned if the two independent runs on the same data set return the same result.
For generating timing results for the brute force method, if there were more than 10,000 simplicies that needed to be generated we instead randomly sampled 10,000 simplices and then scaled the total estimated runtime by a factor of ${n \choose {d+1}} / 10000$.
For generating timing results for Rousseeuw’s method, we ran their algorithm for finding the exact depths of 50 points within $\CD$, and scaled the estimated runtime by a factor of $n/50$.

For \Cref{fig:2dnumPullsvsGap}, 50 simulations were run on the same data set with different random seeds, and the number of pulls per point averaged.

\clearpage
\subsection{Additional Experiments}

Below in \Cref{fig:2dGapHistApp} we provide extended simulation results as in \Cref{fig:2dGapHist}.
\begin{figure}[h]
    \centering
    \begin{subfigure}[b]{0.22\textwidth}
    \includegraphics[width=\textwidth]{figures/gapPlots/n_200_gaps.pdf}
    \subcaption{$n=200$}
    \end{subfigure}
    \begin{subfigure}[b]{0.22\textwidth}
    \includegraphics[width=\textwidth]{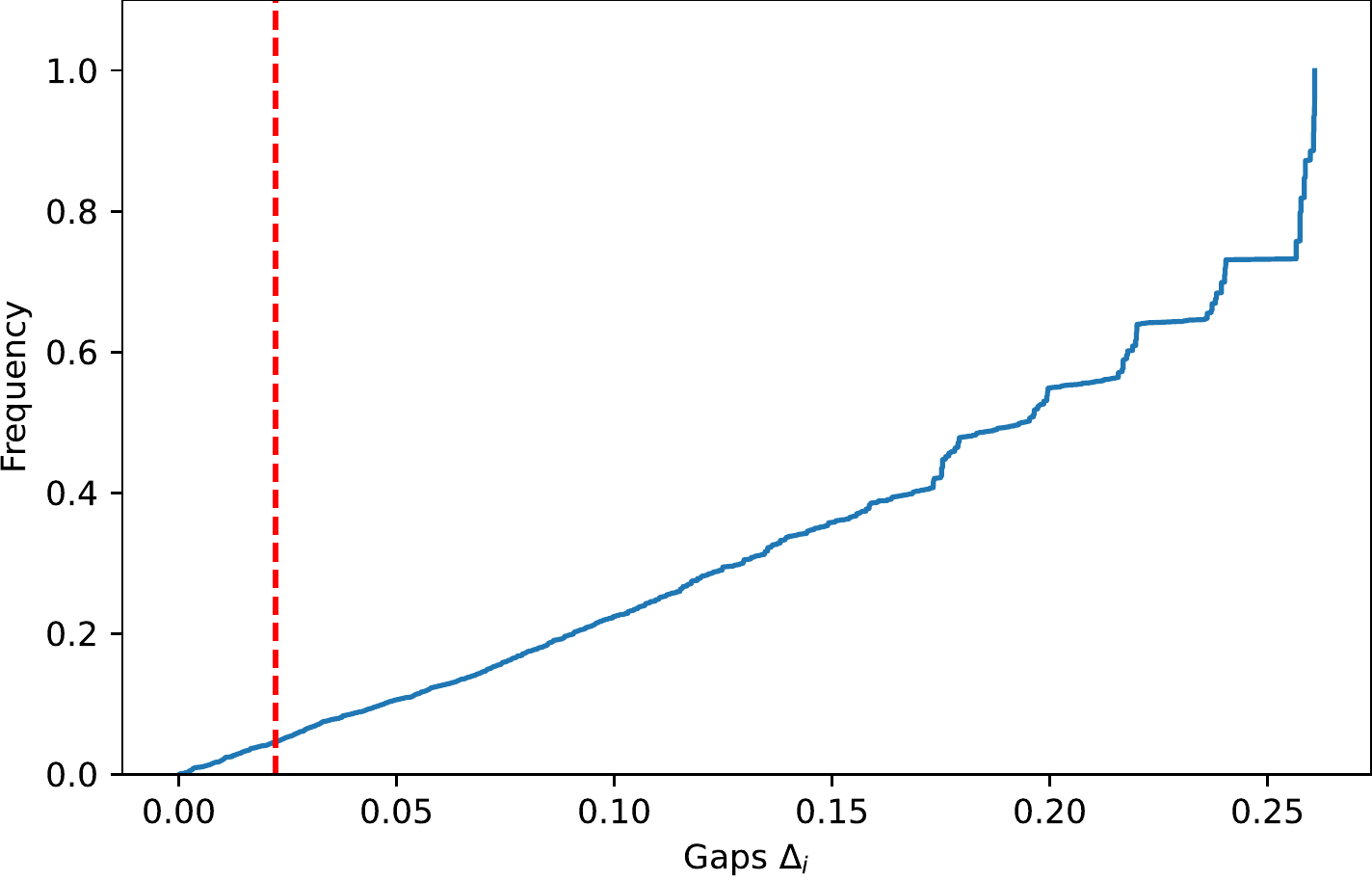}
    \subcaption{$n=255$}
    \end{subfigure}
    \begin{subfigure}[b]{0.22\textwidth}
    \includegraphics[width=\textwidth]{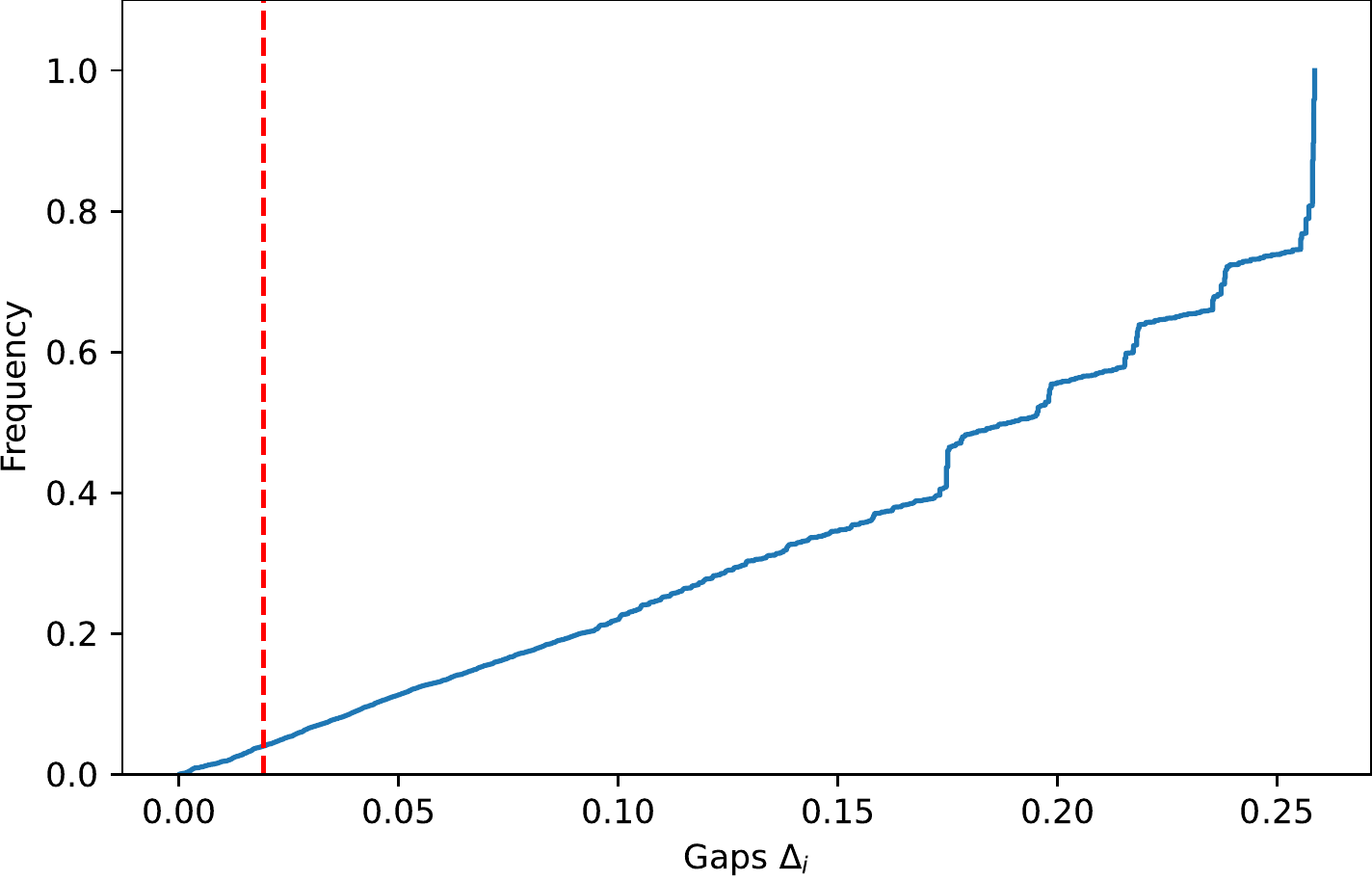}
    \subcaption{$n=325$}
    \end{subfigure}
    \begin{subfigure}[b]{0.22\textwidth}
    \includegraphics[width=\textwidth]{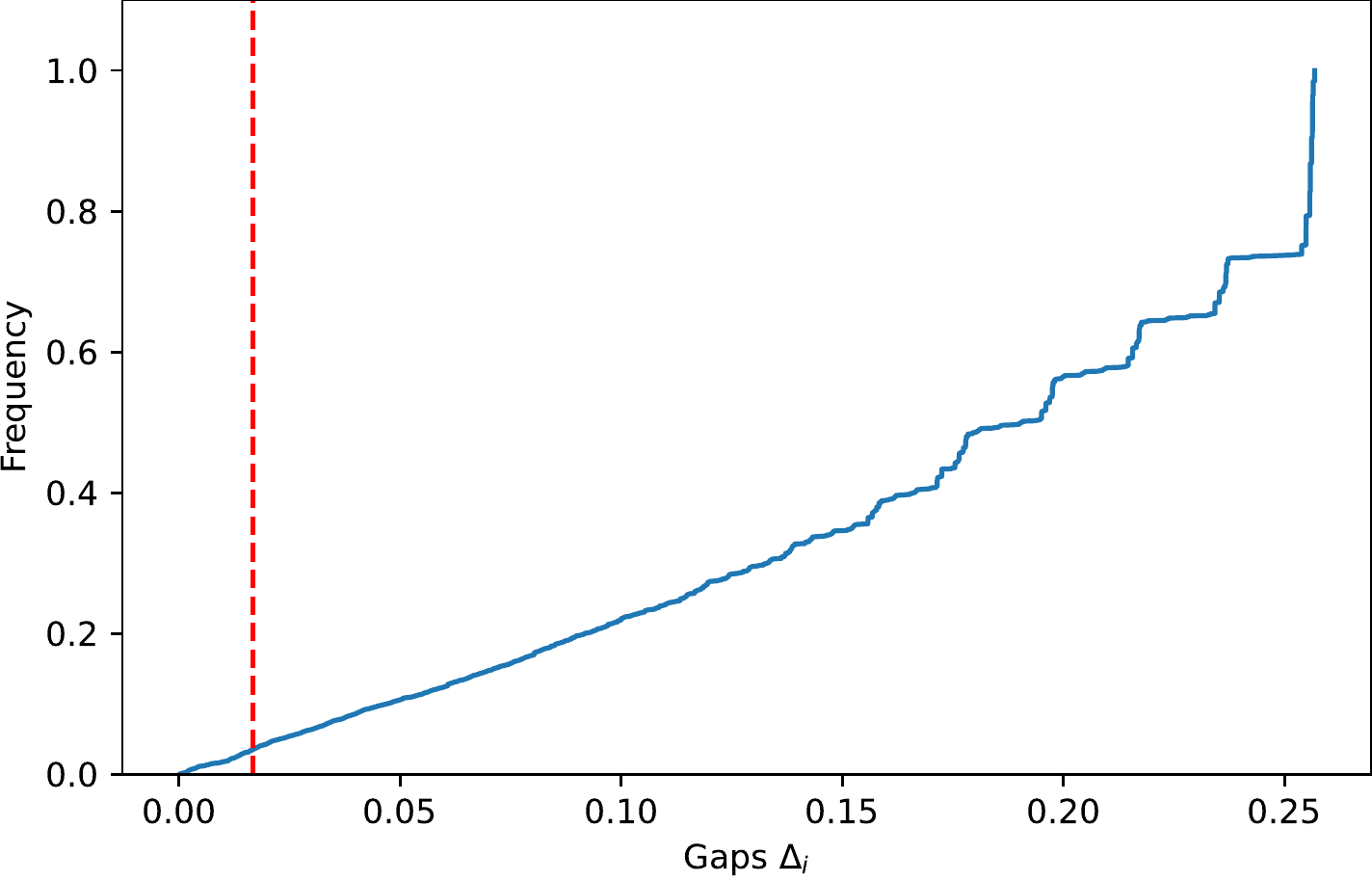}
    \subcaption{$n=414$}
    \end{subfigure}
    \begin{subfigure}[b]{0.22\textwidth}
    \includegraphics[width=\textwidth]{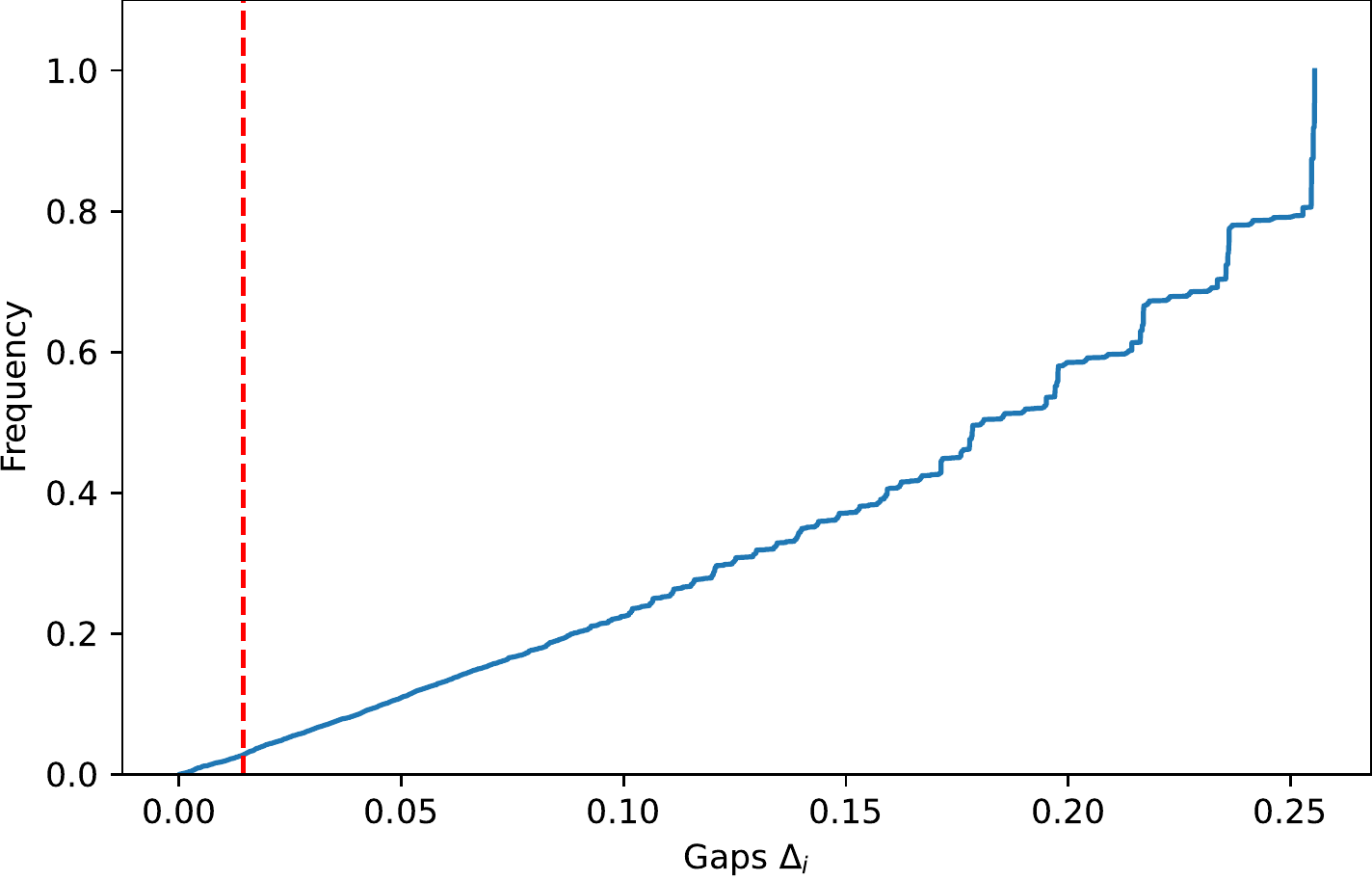}
    \subcaption{$n=527$}
    \end{subfigure}
    \begin{subfigure}[b]{0.22\textwidth}
    \includegraphics[width=\textwidth]{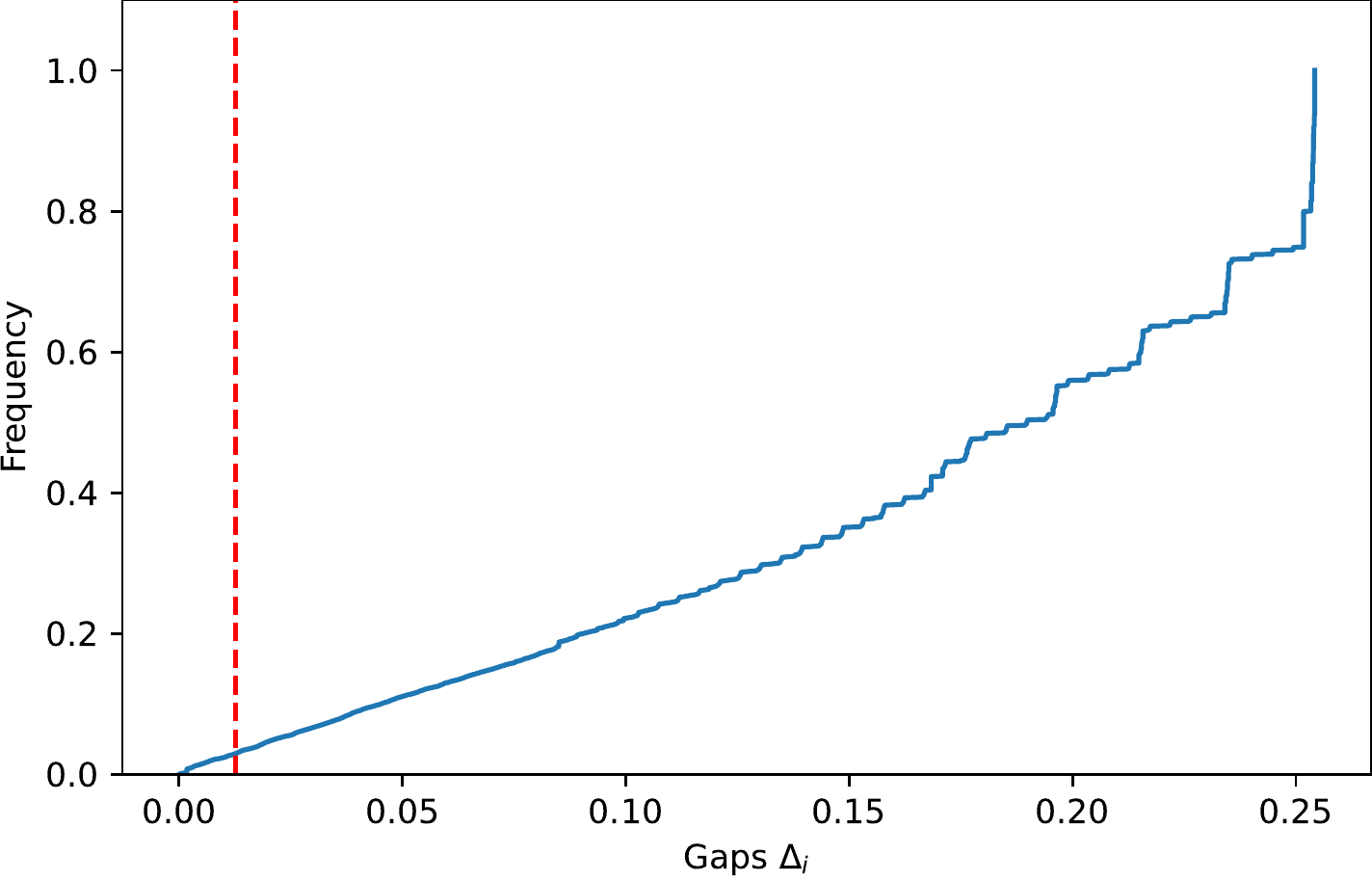}
    \subcaption{$n=672$}
    \end{subfigure}
    \begin{subfigure}[b]{0.22\textwidth}
    \includegraphics[width=\textwidth]{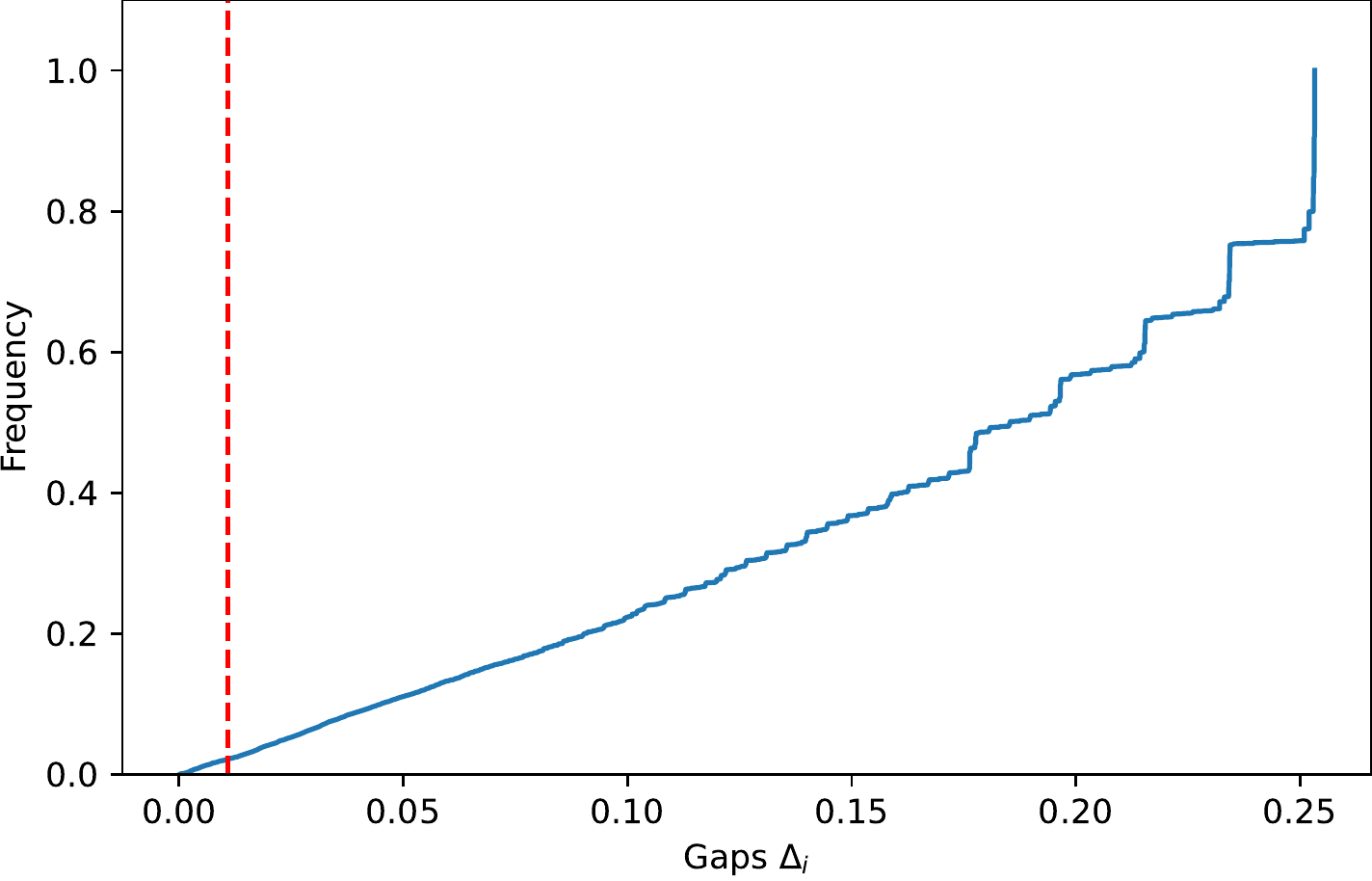}
    \subcaption{$n=856$}
    \end{subfigure}
    \begin{subfigure}[b]{0.22\textwidth}
    \includegraphics[width=\textwidth]{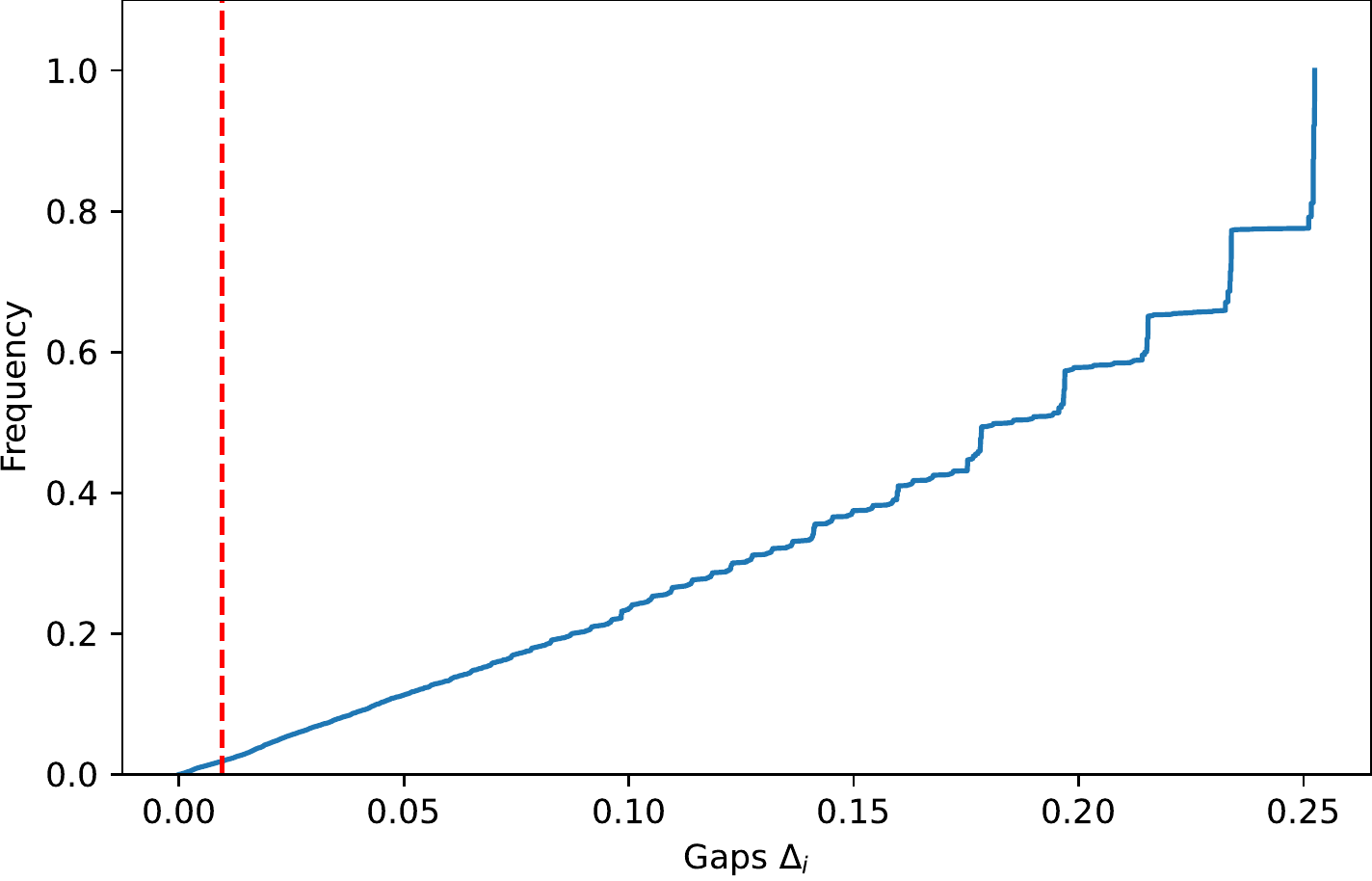}
    \subcaption{$n=1091$}
    \end{subfigure}
    \begin{subfigure}[b]{0.22\textwidth}
    \includegraphics[width=\textwidth]{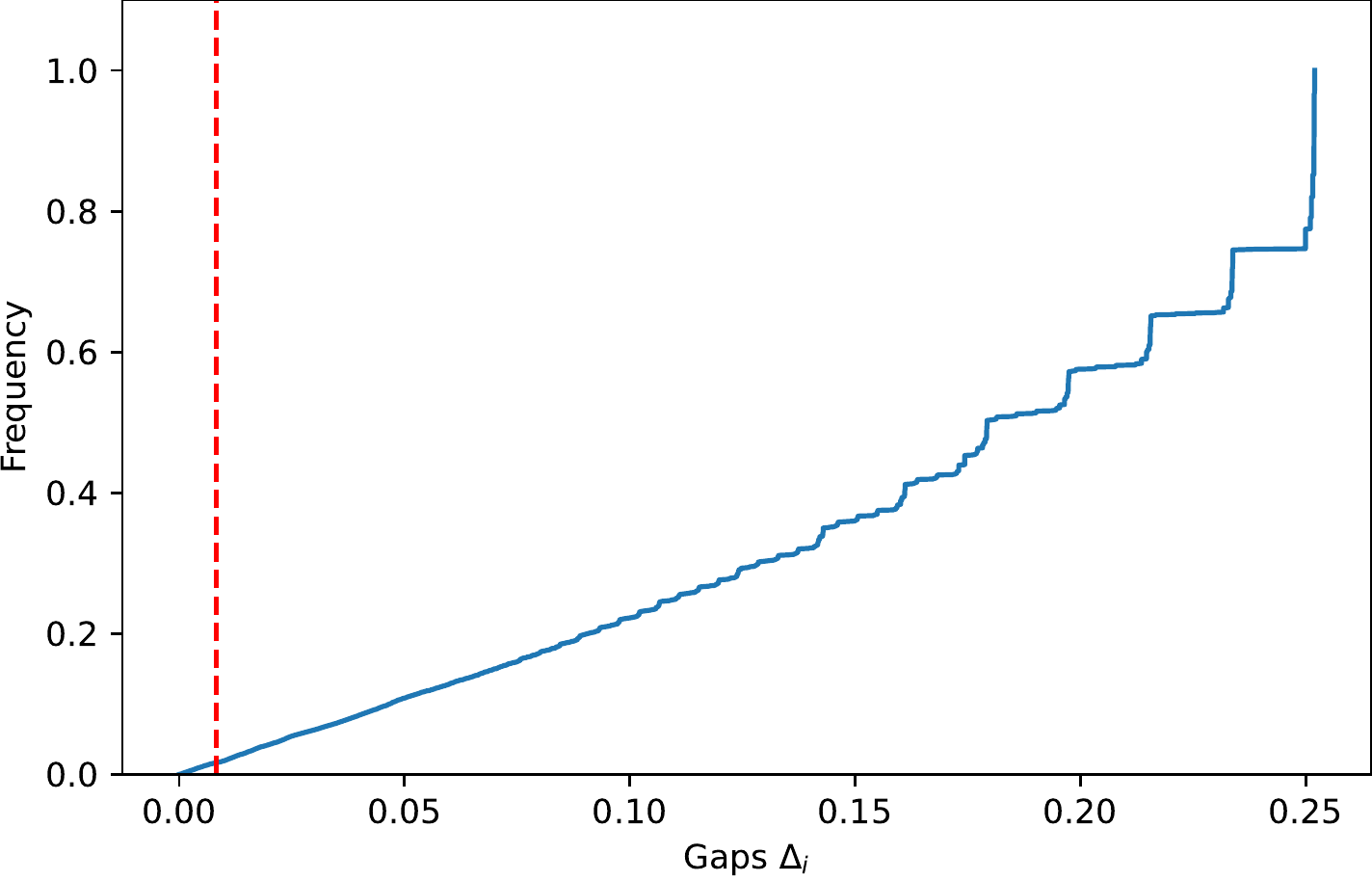}
    \subcaption{$n=1390$}
    \end{subfigure}
    \begin{subfigure}[b]{0.22\textwidth}
    \includegraphics[width=\textwidth]{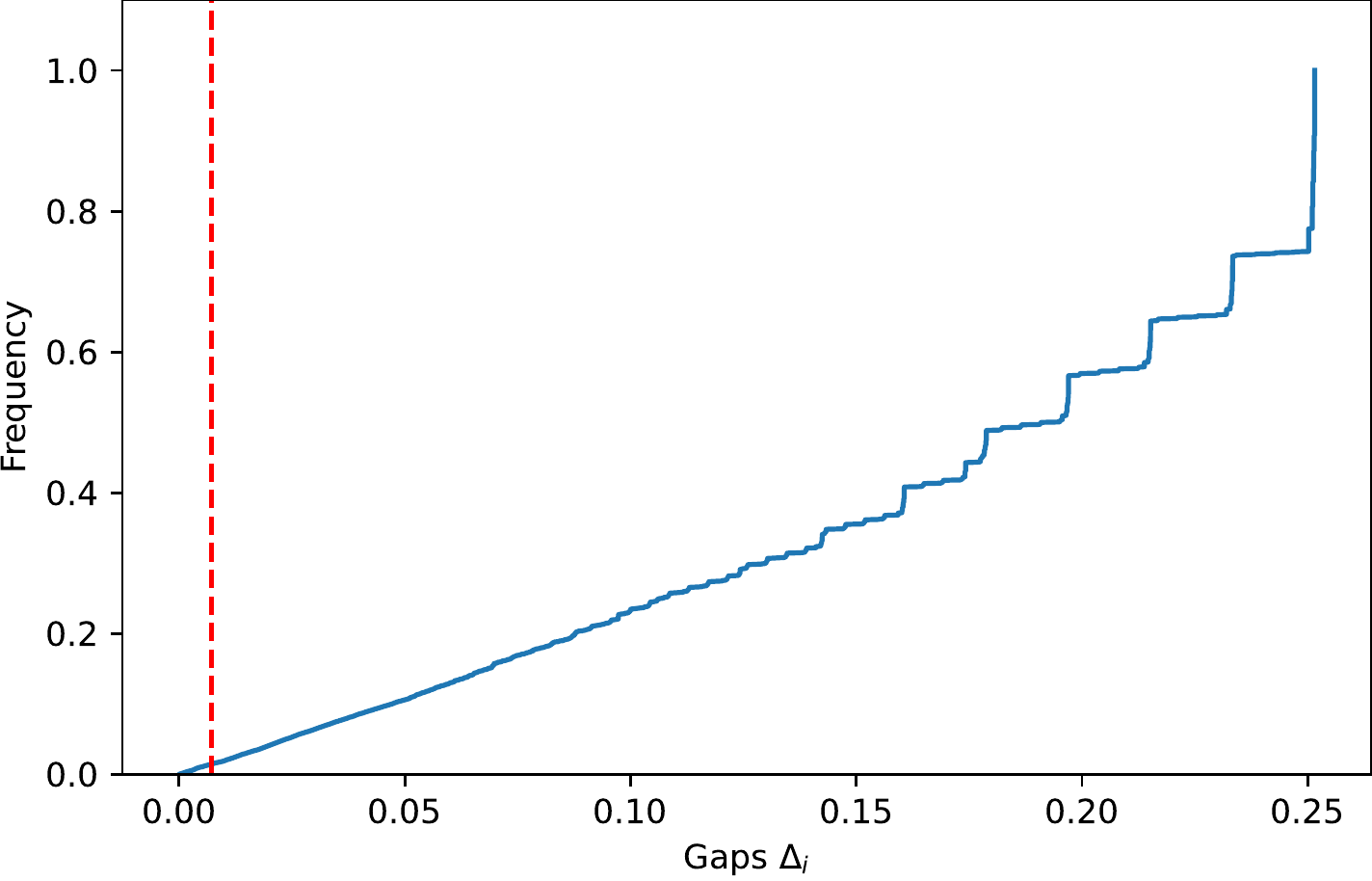}
    \subcaption{$n=1772$}
    \end{subfigure}
    \begin{subfigure}[b]{0.22\textwidth}
    \includegraphics[width=\textwidth]{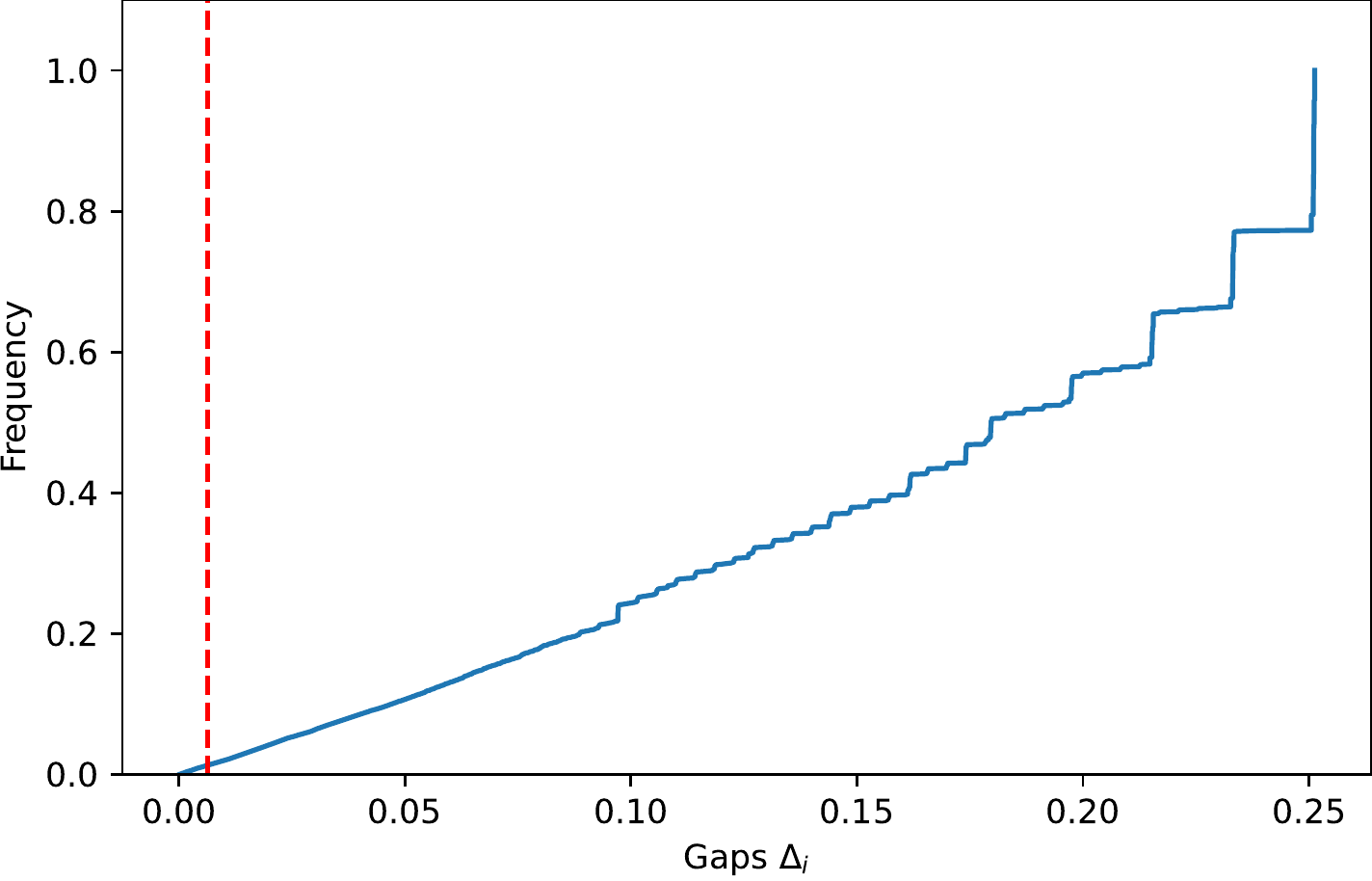}
    \subcaption{$n=2258$}
    \end{subfigure}
    \begin{subfigure}[b]{0.22\textwidth}
    \includegraphics[width=\textwidth]{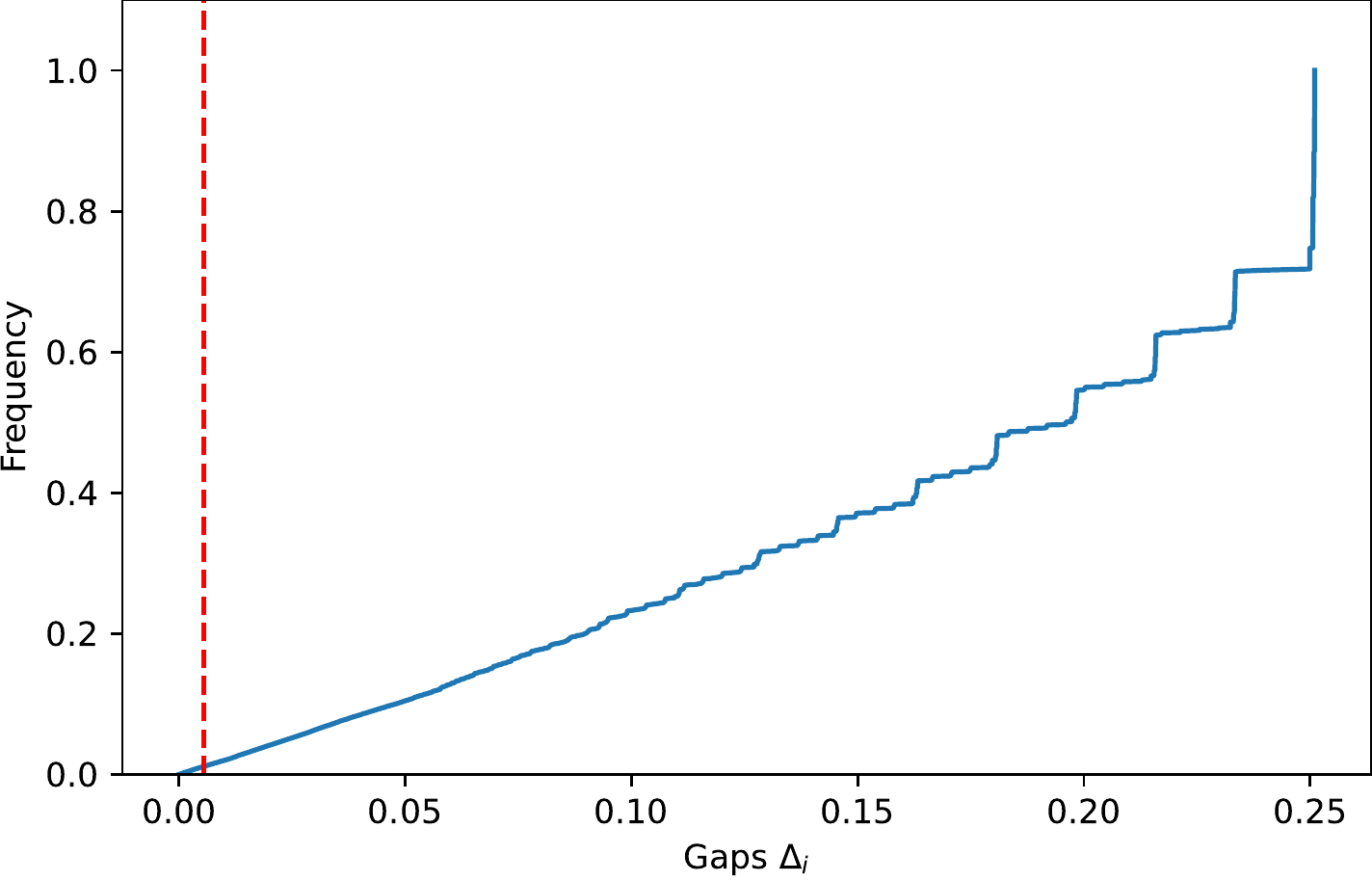}
    \subcaption{$n=2877$}
    \end{subfigure}
    \begin{subfigure}[b]{0.22\textwidth}
    \includegraphics[width=\textwidth]{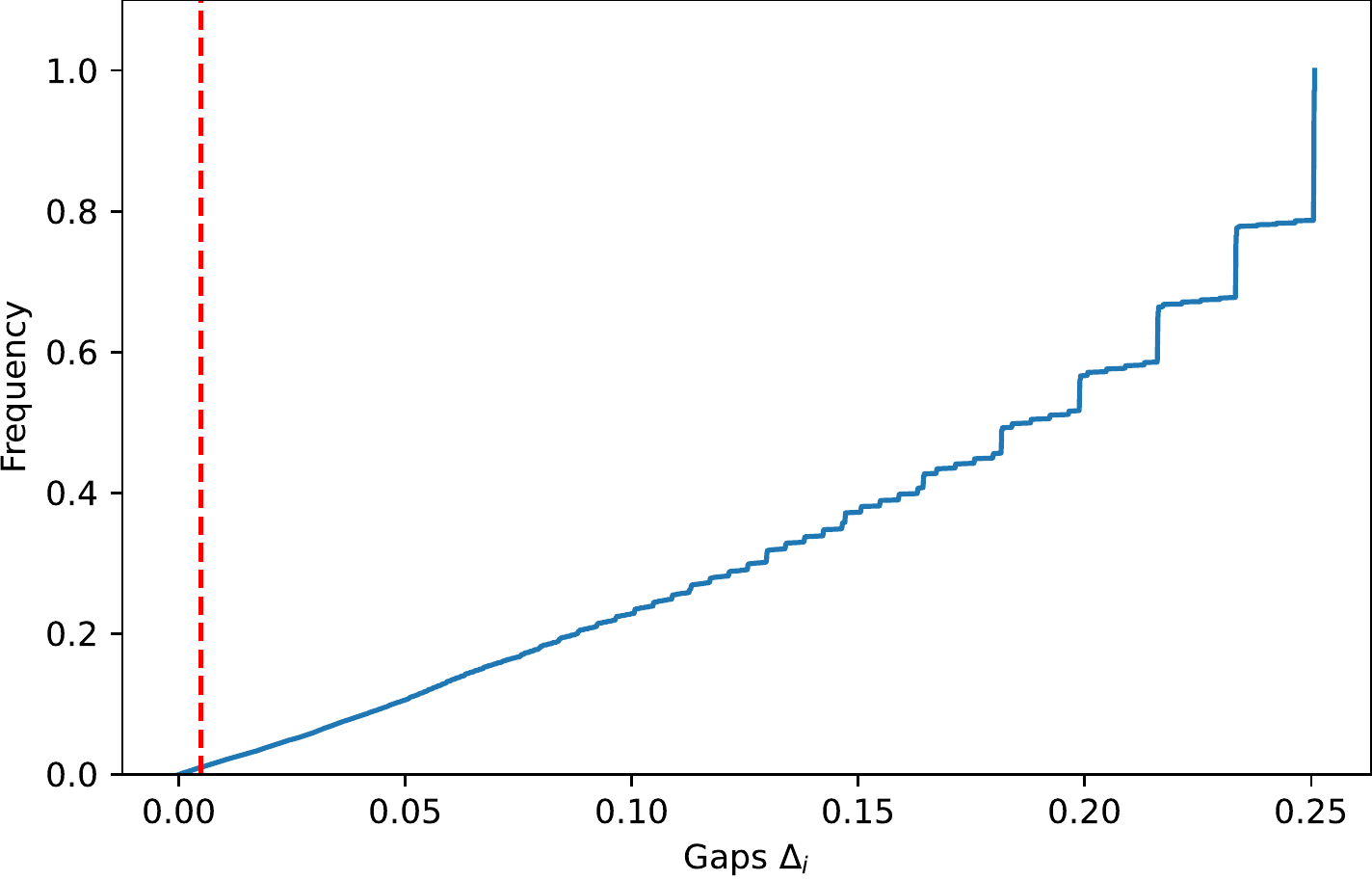}
    \subcaption{$n=3666$}
    \end{subfigure}
    \begin{subfigure}[b]{0.22\textwidth}
    \includegraphics[width=\textwidth]{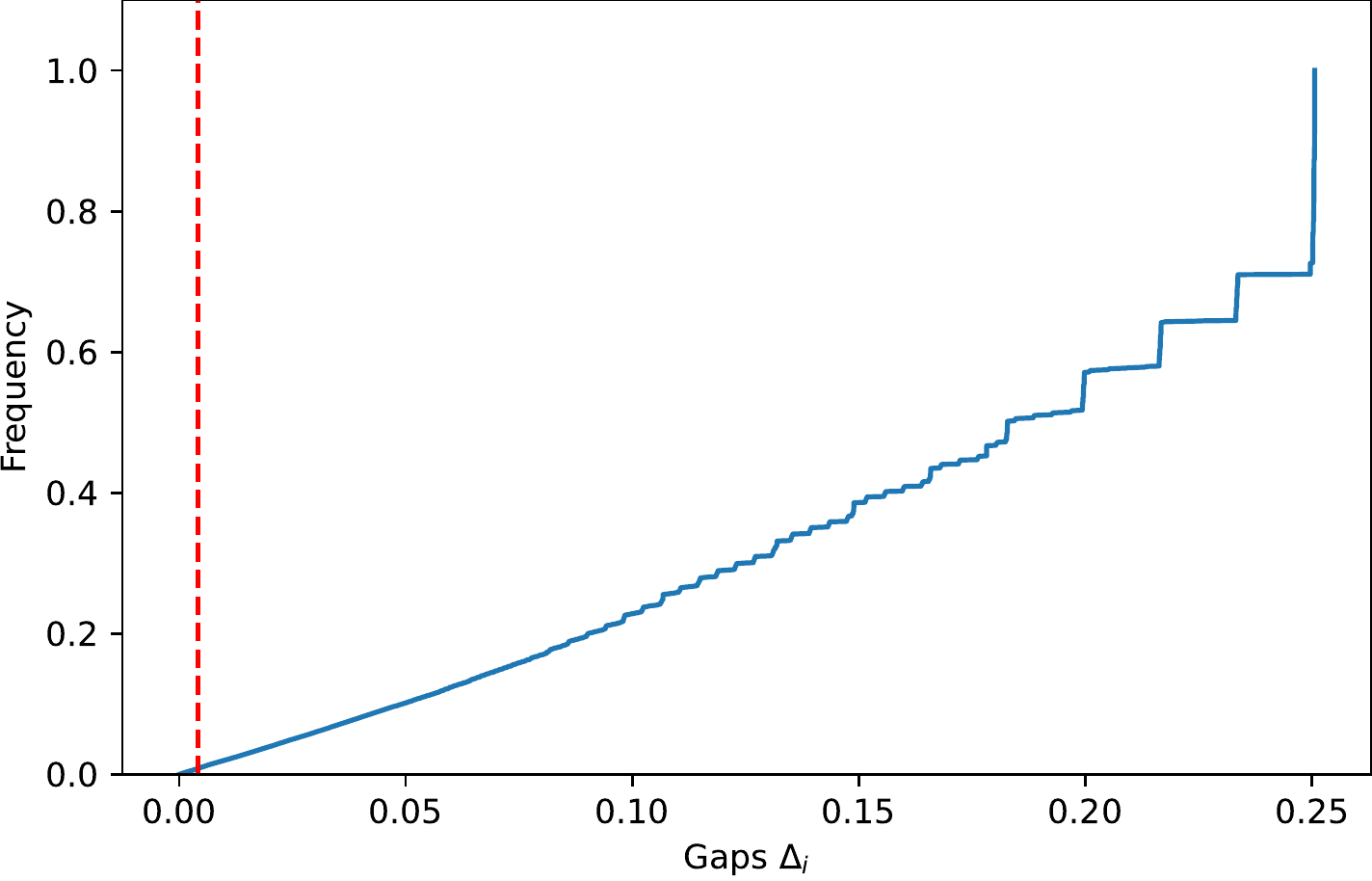}
    \subcaption{$n=4671$}
    \end{subfigure}
    \begin{subfigure}[b]{0.22\textwidth}
    \includegraphics[width=\textwidth]{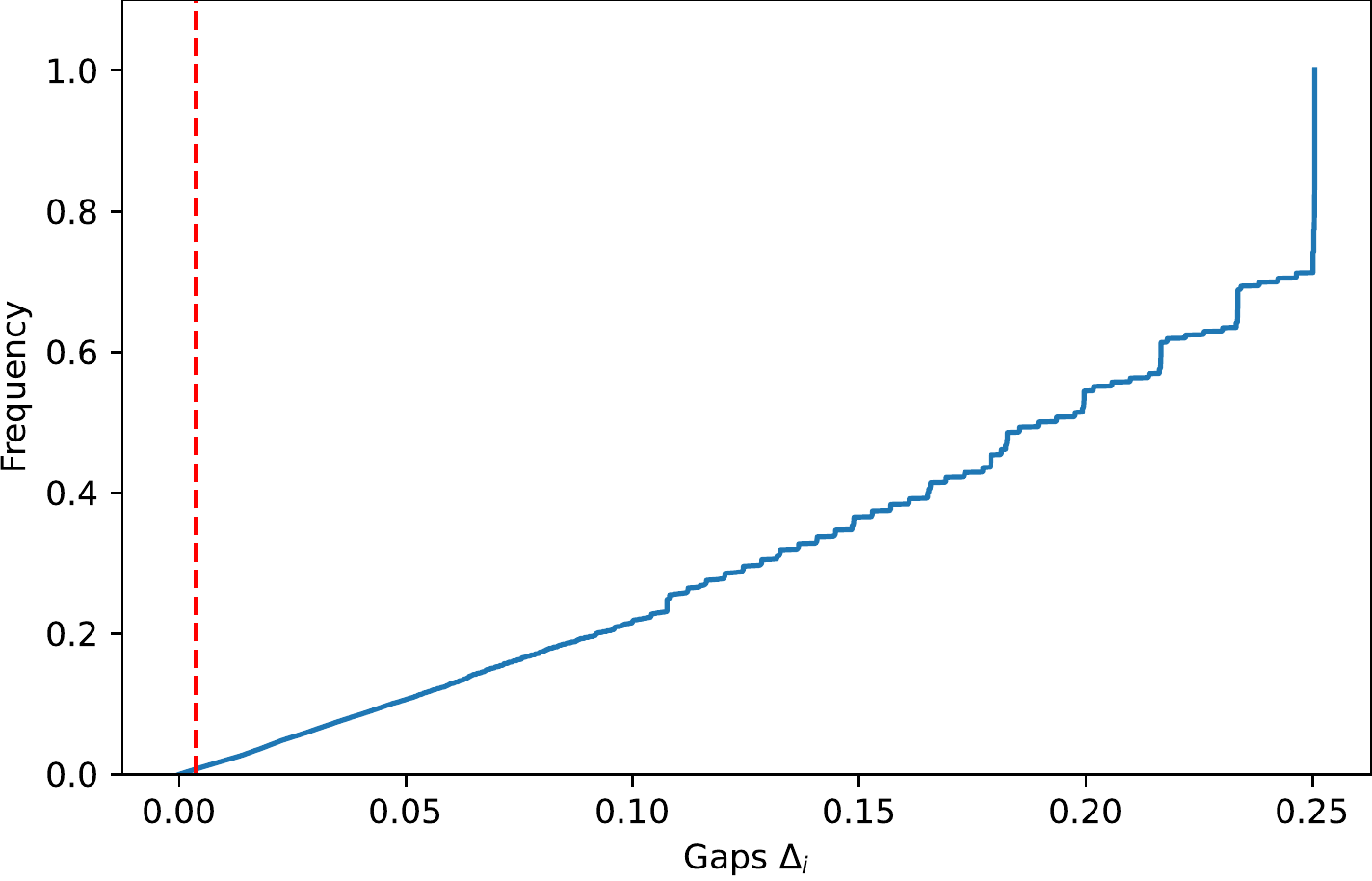}
    \subcaption{$n=5953$}
    \end{subfigure}
    \begin{subfigure}[b]{0.22\textwidth}
    \includegraphics[width=\textwidth]{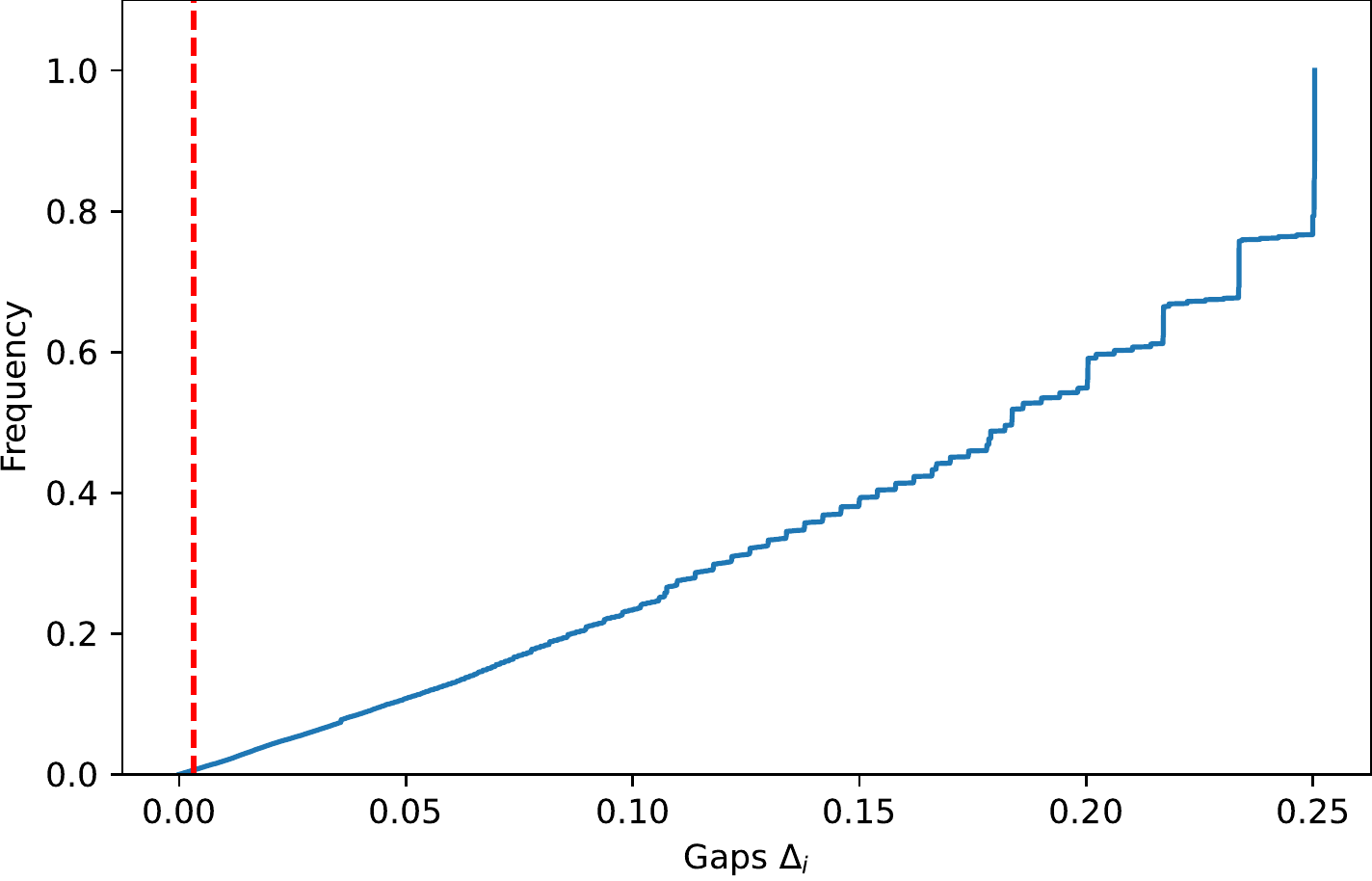}
    \subcaption{$n=7585$}
    \end{subfigure}
    \begin{subfigure}[b]{0.22\textwidth}
    \includegraphics[width=\textwidth]{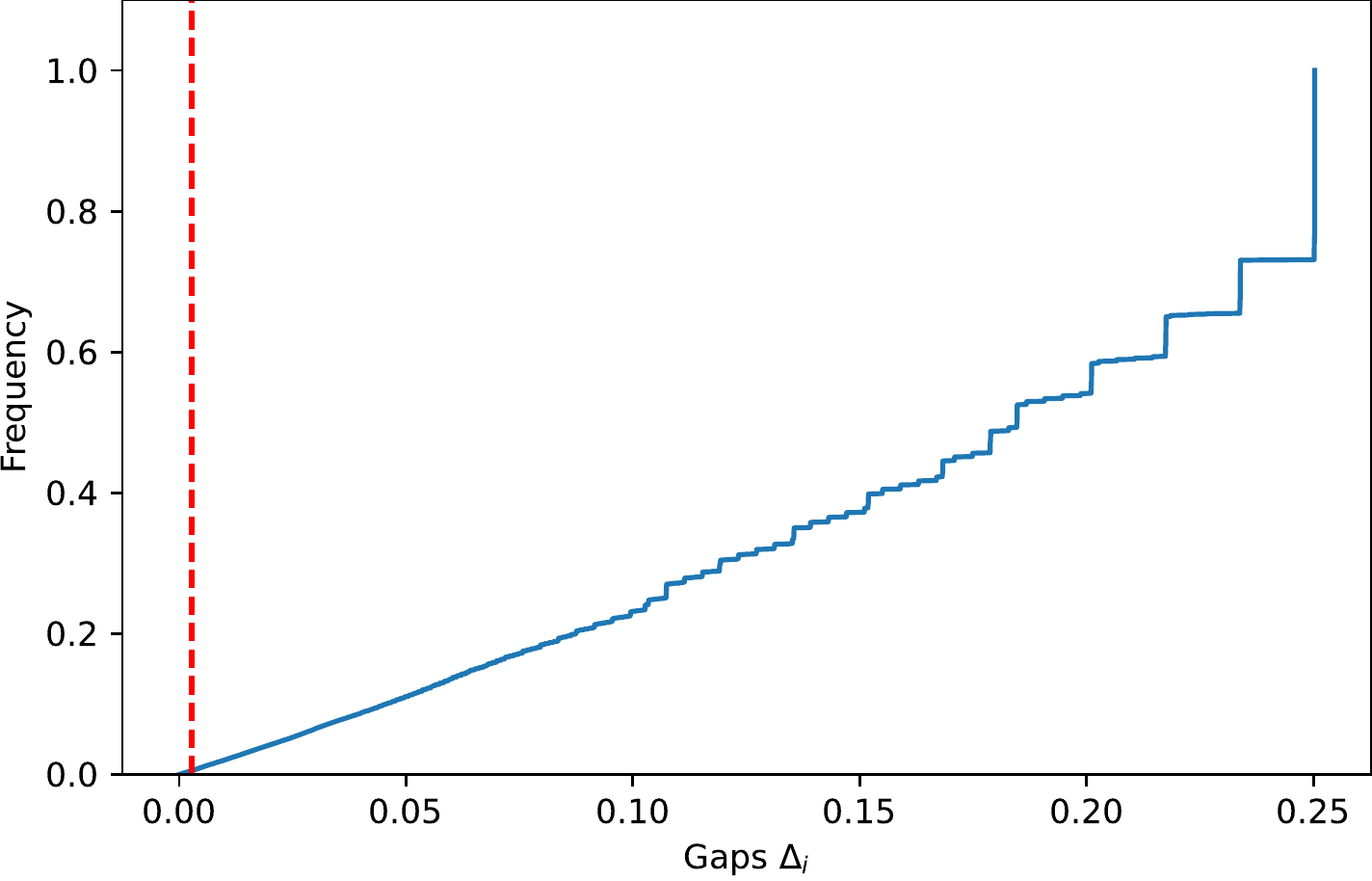}
    \subcaption{$n=9666$}
    \end{subfigure}
    \begin{subfigure}[b]{0.22\textwidth}
    \includegraphics[width=\textwidth]{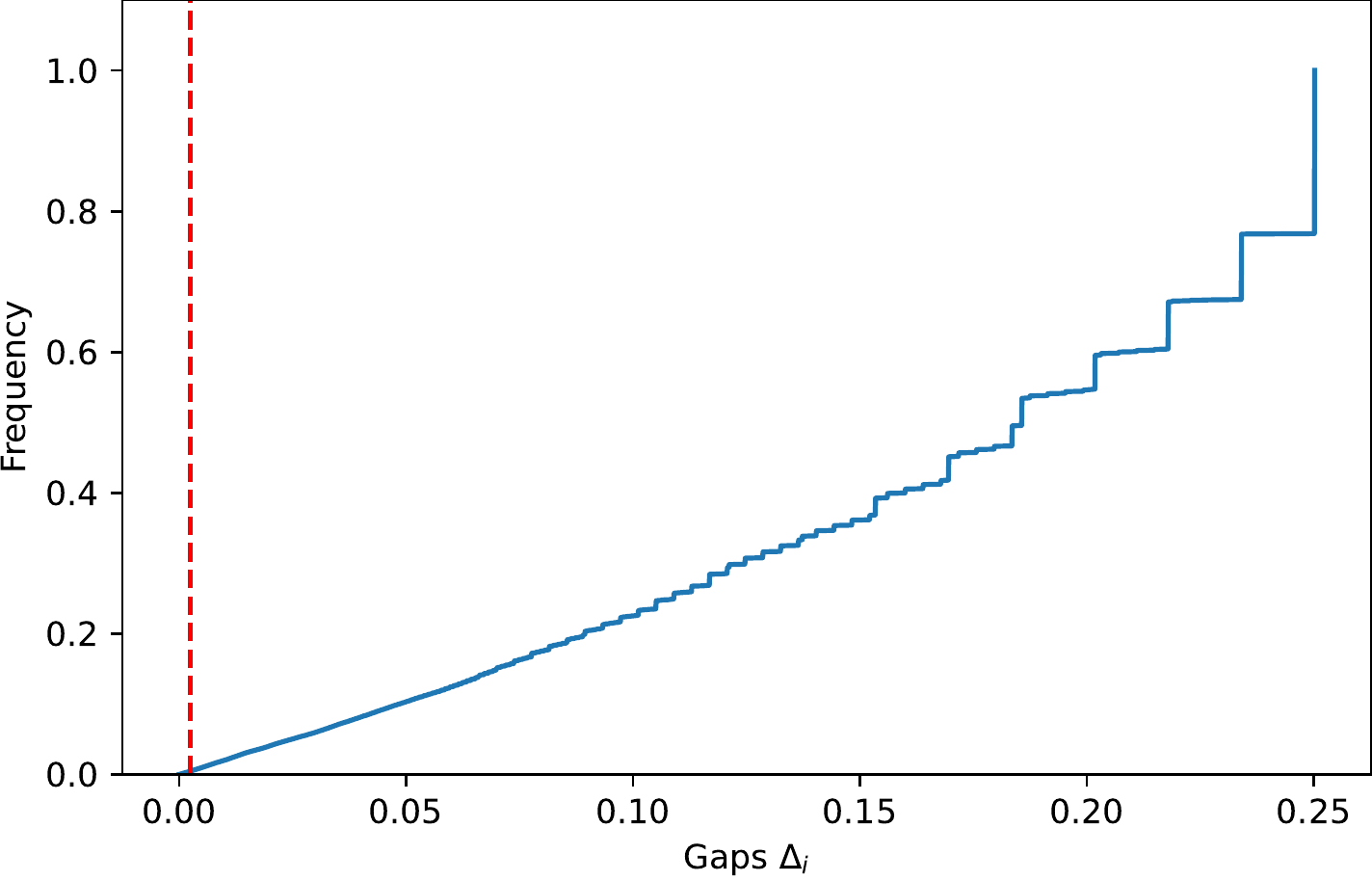}
    \subcaption{$n=12317$}
    \end{subfigure}
    \begin{subfigure}[b]{0.22\textwidth}
    \includegraphics[width=\textwidth]{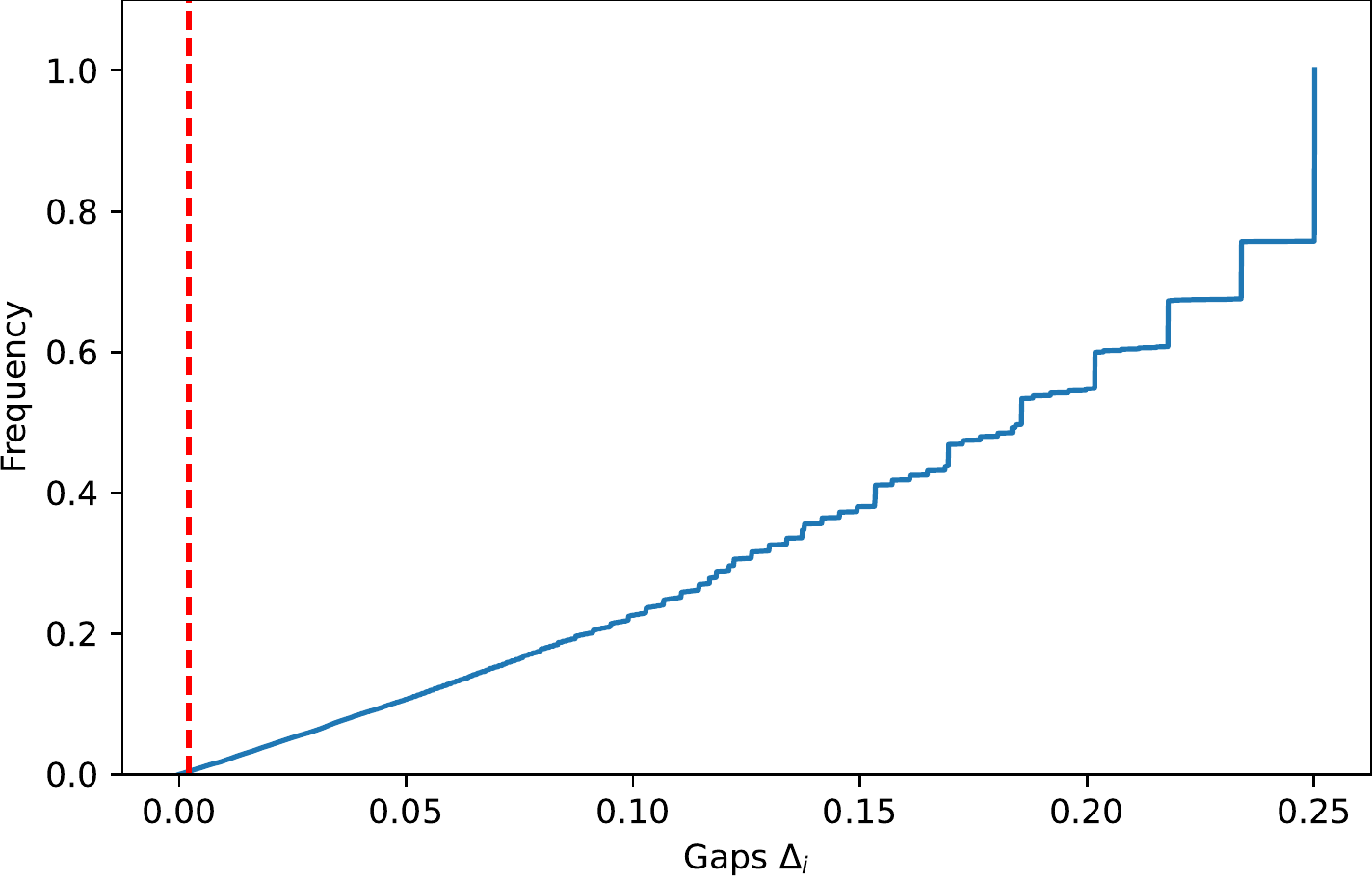}
    \subcaption{$n=15695$}
    \end{subfigure}
    \begin{subfigure}[b]{0.22\textwidth}
    \includegraphics[width=\textwidth]{figures/gapPlots/n_20000_gaps.pdf}
    \subcaption{$n=20000$}
    \end{subfigure}
    \caption{Additional gap plots, as in \Cref{fig:2dGapHist}.}\label{fig:2dGapHistApp}
\end{figure}

\clearpage
\bibliography{refs}

\end{document}